\documentclass[%
 reprint,
 superscriptaddress,
 nofootinbib,
 amsmath,amssymb,
 aps,
 prx,
]{revtex4-2}

\usepackage[utf8]{inputenc}
\usepackage[english]{babel}
\usepackage{tikz}
\usetikzlibrary{shapes,backgrounds,fit,decorations.pathreplacing,arrows,decorations.markings}
\usepackage{verbatim}
\tikzstyle{vecArrow} = [thick, decoration={markings,mark=at position
   1 with {\arrow[semithick]{open triangle 60}}},
   double distance=1.4pt, shorten >= 5.5pt,
   preaction = {decorate},
   postaction = {draw,line width=1.4pt, white,shorten >= 4.5pt}]
\tikzstyle{innerWhite} = [semithick, white,line width=1.4pt, shorten >= 4.5pt]
\usepackage{bbm}
\usepackage{bm}
\usepackage{amsbsy}
\usepackage{amssymb}
\usepackage{amsfonts}
\usepackage{amsmath}
\usepackage{dsfont} 
\usepackage{graphicx} 
\usepackage{epsfig}
\usepackage{epstopdf}
\usepackage{color}

\usepackage[colorlinks,linkcolor=blue,urlcolor=blue, citecolor=blue]{hyperref}

\makeatletter
\newcommand\org@hypertarget{}
\let\org@hypertarget\hypertarget
\renewcommand\hypertarget[2]{%
  \Hy@raisedlink{\org@hypertarget{#1}{}}#2%
  }
\makeatother
\usepackage[figure,table]{hypcap}
\usepackage{enumerate}
\usepackage{enumitem}
\usepackage{float}
\definecolor{darkgreen}{RGB}{50,190,50}
\definecolor{darkblue}{RGB}{0,0,190}
\definecolor{darkred}{RGB}{238,0,0}
\definecolor{quantum}{RGB}{83,37,127}
\definecolor{quantumlight}{RGB}{169,146,191}
\usepackage{soul}

\newcommand{\ket}[1]{\ensuremath{\left|\right.\!{#1}\!\left.\right\rangle}}

\newcommand{\braket}[2]{\ensuremath{\langle{#1}|{#2}\rangle}}
\newcommand{\ketbra}[2]{\ensuremath{|{#1}\rangle\!\langle{#2}|}}

\newcommand{\tr}{\textnormal{Tr}}

\usepackage{amsthm} 
\usepackage{thmtools} 
\usepackage{thm-restate}

\declaretheorem[]{corollary,lemma,proposition,definition}

\newtheorem*{hypothesis*}{Past Hypothesis}
\newtheorem*{main*}{Equilibrium entropy}

\makeatother

\begin{document}


\title{Emergence of a second law of thermodynamics in isolated quantum systems}

\author{Florian Meier}
\email[]{florianmeier256@gmail.com}
\affiliation{Atominstitut, Technische Universit{\"a}t Wien, 1020 Vienna, Austria}

\author{Tom Rivlin}
\email{tom.rivlin@tuwien.ac.at}
\affiliation{Atominstitut, Technische Universit{\"a}t Wien, 1020 Vienna, Austria}

\author{Tiago Debarba}
\affiliation{Atominstitut, Technische Universit{\"a}t Wien, 1020 Vienna, Austria}
\affiliation{Departamento Acad{\^ e}mico de Ci{\^ e}ncias da Natureza, Universidade Tecnol{\'o}gica Federal do Paran{\'a} (UTFPR), Campus Corn{\'e}lio Proc{\'o}pio, Avenida Alberto Carazzai 1640, Corn{\'e}lio Proc{\'o}pio, Paran{\'a} 86300-000, Brazil.}

\author{Jake Xuereb}
\affiliation{Atominstitut, Technische Universit{\"a}t Wien, 1020 Vienna, Austria}

\author{Marcus Huber}
\email{marcus.huber@tuwien.ac.at} 
\affiliation{Atominstitut, Technische Universit{\"a}t Wien, 1020 Vienna, Austria}
\affiliation{Institute for Quantum Optics and Quantum Information - IQOQI Vienna, Austrian Academy of Sciences, Boltzmanngasse 3, 1090 Vienna, Austria}

\author{Maximilian P.~E. Lock}
\email{maximilian.paul.lock@tuwien.ac.at} 
\affiliation{Atominstitut, Technische Universit{\"a}t Wien, 1020 Vienna, Austria}
\affiliation{Institute for Quantum Optics and Quantum Information - IQOQI Vienna, Austrian Academy of Sciences, Boltzmanngasse 3, 1090 Vienna, Austria}

\date{\today}

\begin{abstract}
The second law of thermodynamics states that the entropy of an isolated system can only increase over time.
This appears to conflict with the reversible evolution of isolated quantum systems under the Schrödinger equation, which preserves the von Neumann entropy.
Nonetheless, one finds that with respect to many observables, expectation values approach a fixed value -- their equilibrium value.
This ultimately raises the question: \textit{in what sense does the entropy of an isolated quantum system increase over time}?
For classical systems, one introduces the assumption of a low entropy initial state along with the concept of ignorance about the microscopic details of the physical system, leading to a statistical interpretation of the second law.
By considering the observables through which we examine quantum systems, both these assumptions can be incorporated, building upon recent studies of the \textit{equilibration on average} of observables.
While the statistical behavior of observable expectation values is well-established, a quantitative connection to entropy increase has been lacking so far.
In deriving novel bounds for the equilibration of observables, and considering the entropy of the system relative to observables, we recover a variant of the second law: the entropy with respect to a given observable tends towards its equilibrium value in the course of the system's unitary evolution.
These results also support recent findings which question the necessity of non-integrability for equilibration in quantum systems.
We further illustrate our bounds using numerical results from the paradigmatic example of a quantum Ising model on a chain of spins. There, we observe entropy increasing up to equilibrium values, as well as fluctuations which expose the underlying reversible evolution in accordance with the derived bounds.
\end{abstract}

\maketitle

\raggedbottom

\section{\label{sec:introduction}Introduction}
Irreversible evolution observed in nature is characterized by the second law of thermodynamics~\cite{Lieb1999,Thess2011,Ng2018,Kammerlander2019}, and is thus quantitatively captured by the increase in entropy over time of an isolated system.
From the perspective of microscopic physics, however, all laws are time-reversal invariant and, thus, if a process is allowed by these laws, its reverse process must be, too. This means that entropy could, in principle, also decrease.
This apparent conflict between the reversibility of physics on the microscopic level and the irreversibility of macroscopic physics, sometimes called \textit{Loschmidt's paradox}, has been debated since the earliest days of classical statistical mechanics~\cite{zwanzig1970concept,steckline1983zermelo,brown2009boltzmann,Wehrl1978}.
As for the entropy, for classical systems, Liouville's theorem implies that the Gibbs entropy formula is constant under Hamiltonian evolution~\cite{callender1999reducing},
and also for quantum systems, the unitary evolution generated by the Schr\"{o}dinger equation preserves the von Neumann entropy.
In this article, we address this issue from the perspective of quantum systems, determining the sense in which entropy can increase over time in isolated quantum systems.

Historically, resolutions to this conflict for classical systems have typically involved two components. The first component is an assumption of a low-entropy initial state, the \textit{past hypothesis}~\cite{boltzmann1964,albert2000time}. 
The second component corresponds to a loss or lack of information~\cite{jaynes1957information,jaynes1957information2,zwanzig1970concept}, for example via the \textit{Stosszahlansatz} in Boltzmann's H-theorem~\cite{brown2009boltzmann}, or by a coarse-graining of the phase space~\cite{jaynes1957information,jaynes1957information2,isihara2013statistical}.
The modern consensus is that, while the microscopic physics is reversible, the second law of thermodynamics holds on average, meaning that for most of the time, an isolated system looks the same -- like the so-called  \textit{equilibrium average}.
Decreases in entropy occur with a probability which vanishes exponentially with the magnitude of this decrease, as captured by the celebrated fluctuation relations~\cite{evans1993probability,crooks1998nonequilibrium,wang2002experimental,penrose2005foundations,grandy2008entropy,Parrondo2015}.

With regards to quantum mechanics, despite substantial theoretical arguments~\cite{tasaki1998quantum,Polkovnikov2011,Ikeda2015,Reimann2012,Short_2012,gogolin2016equilibration,Lent2019,Strasberg2021,Teufel2023,Nagasawa2024} and experimental evidence~\cite{trotzky2012probing,langen2013local,Langen2015,trotzky2012probing} for the equilibration of isolated quantum systems, no consensus exists about \textit{how} isolated quantum systems equilibrate in the sense of something akin to the second law, as unitary evolution preserves the von Neumann entropy.
Progress has been made in showing that, for example, the diagonal entropy of closed, driven (i.e. energy non-conserving) quantum systems at some time $t>0$ is non-decreasing from its value at $t=0$~\cite{Polkovnikov2011,Ikeda2015}, or that the observational entropy of most states in the Hilbert space is close to the maximum~\cite{Nagasawa2024}.
However, a proof that the entropy of an isolated system tends to an equilibrium value over time, and a general exposition of the conditions under which this occurs, has been missing.
We fill this gap using the modern theory of equilibration of isolated quantum systems~\cite{Reimann2012,Short_2012} showing how the second law of thermodynamics is satisfied \textit{on average}.

In this framework, the system is examined relative to an arbitrary observable represented by a self-adjoint operator $O=\sum_{i=1}^r O_i \Pi_i^O$, where $r$ is the number of possible measurement outcomes.
Assuming finite Hilbert space dimension $d$ (with $r<d$) and identifying the measurement outcomes with macrostates introduces, in a sense which we will quantify, the necessary lack of information, and ensuing emergent irreversibility.

When we access the state $\rho$ via the observable $O$, the available information about the system is represented by the probability distribution $\vec p = (p_1,\dots,p_r),$ where $p_i = \tr\left[\Pi_i^O \rho\right]$ are the populations of $\rho$ with respect to the orthogonal projectors $\Pi_i$ of $O$.
In place of the usual von Neumann entropy $S_\mathrm{vN}[\rho]=-\tr[\rho \log \rho]$, which is invariant under the unitary evolution generated by the Schr\"{o}dinger equation $\dot \rho(t) = -i[H,\rho(t)]$ with time-independent Hamiltonian $H$, the information obtained about the system is then quantified by the Shannon entropy of the observable's probability distribution -- its \textit{Shannon observable entropy} 
\begin{equation} \label{eq:ShanObsEnt}
    S_{\mathrm{Sh}}^O[\rho] = -\sum_{i=1}^r p_i \log p_i,
\end{equation}
as in~\cite{Ingarden1976}.
Note that instead of the projective valued measure (PVM) associated with the observable $O$ as we use it in the main text, one can define the probability distribution $\vec p$, with respect to a positive operator valued measure (POVM) $(E_i)_{i=1}^r$ and all the following results follow analogously (see Appendix~\ref{appendix:replacingPOVMs}).

Unlike the von Neumann entropy, the Shannon observable entropy can change in time for isolated systems. Indeed, we find that over the course of the evolution of an isolated system, the Shannon observable entropy is close to that of the infinite-time-averaged state
\begin{align}
\label{eq:def_omega}
    \omega=\lim_{T\rightarrow\infty}\frac{1}{T}\int_0^T dt \rho(t),
\end{align}
also known as the equilibrium state in the context of equilibration on average:
\begin{restatable}[]{main*}{mainresult}
    With $\rho(t),$ $\omega$ and $O$ as defined on a finite dimensional Hilbert space, we have
    \begin{align}
    \label{eq:shannon_entropy_intro}
        \left\langle \left |S_{\mathrm{Sh}}^O[\rho(t)]-S_{\mathrm{Sh}}^O[\omega]\right|\right\rangle_\infty \leq \mathcal O\left(\log(r) \sqrt{\frac{r}{d_\mathrm{eff}}}\right),
    \end{align}
    where the average $\langle\cdot\rangle_\infty$ is with respect to uniformly random times $0\leq t<\infty$, and the right-hand side (expressed in terms of big-$\mathcal O$ notation) vanishes in the limit where the number of microstates $d_\mathrm{eff}=1/\tr[\omega^2]$ is much larger than that of macrostates $r.$
\end{restatable}
The small parameter on the right-hand side of~\eqref{eq:shannon_entropy_intro} measures how strongly the observable $O$ coarse-grains the system relative to the number of microstates $d_\mathrm{eff}$, as discussed in more detail later.
Here we combine improvements for bounds from the theory of equilibration on average~\cite{Reimann2012,Short_2012,gogolin2016equilibration}, with the so-called Past Hypothesis, formally stated as:
\begin{restatable}[]{hypothesis*}{pasthypothesis}
    The assumption that we initially start in a low entropy state with respect to $O$,
    \begin{align}
    \label{eq:past_hypothesis}
        S_\mathrm{Sh}^O[\rho(0)] \ll S_\mathrm{Sh}^O[\omega].
    \end{align}
\end{restatable}
This results in the following statement of the second law in unitary quantum mechanics: \textit{the Shannon observable entropy of an isolated system increases towards its equilibrium value}.

\section{\label{sec:formalism}Formalism}
\paragraph*{Equilibrating observables.}

We will work in the setting where the dynamics of our system of interest is described by a Hamiltonian with spectral decomposition,
\begin{align}
    \label{eq:H_decomposition}
    H = \sum_{\lambda\in\sigma(H)} \lambda \Pi_\lambda,
\end{align}
where the sum runs over the spectrum $\lambda\in\sigma(H)$ of the Hamiltonian and $\Pi_\lambda$ are the orthogonal projectors onto the eigenspaces corresponding to the eigenvalues $\lambda$ of $H$.
The dimension of the system's Hilbert space can be recovered via $d=\sum_{\lambda\in\sigma(H)}\tr\left[\Pi_\lambda\right]$ and is assumed to be finite in this work.
Isolated quantum systems evolve according to the von Neumann equation, $\dot\rho(t) = -i\left[H,\rho(t)\right]$, given the initial condition $\rho(0)$. Note that $\hbar = 1$ throughout this work.

While the state of the system at a given instant in time $t$ is of course $\rho(t)$, it was shown in~\cite{tasaki1998quantum,Reimann2012,Short_2012,gogolin2016equilibration} that, for large systems, the expectation value of many observables will be close to that of the equilibrium state $\omega$ for most times.
Such observables are therefore said to equilibrate \emph{on average}. 
To make this quantitative, we use the notation $\langle X(t) \rangle_T = T^{-1}\int_0^T dt X(t)$ to denote the average of an arbitrary quantity $X(t)$ over the finite time window $[0,T]$. 
We can thus write the equilibrium state as $\omega=\lim_{T\rightarrow\infty} \left\langle\rho(t)\right\rangle_{T}$ in agreement with Eq.~\eqref{eq:def_omega}.
Two questions must then be answered: how `large' must the system be and how `close' exactly are the expectation values of $\rho(t)$ to those of $\omega$?

The size of the system can be quantified by its so-called effective dimension~\cite{linden2009quantum}, defined via:
\begin{align}
\label{eq:d_eff}
    \frac{1}{(d_\mathrm{eff})^{2}} =  \sum_{\lambda\in\sigma(H)} \tr [ \Pi_{\lambda} \rho(t) ]^2.
\end{align}
The effective dimension is independent of time $t$, and we can equivalently write $d_\mathrm{eff} = 1/\tr[\omega^2]$. 
This is essentially a measure of how many non-degenerate energy eigenstates participate in the quantum state's evolution.
One can then compare the difference between the observable expectation values. For example, it has been shown~\cite{reimann2008foundation,Short_2012} that
\begin{align}
\label{eq:ReimannBound}
    \left\langle \left|\tr \left[O\rho(t)\right]- \tr\left[O\omega\right]\right|^2 \right\rangle_{T} \leq \frac{\|O\|^{2}}{d_\mathrm{eff}} f(\varepsilon,T),
\end{align} 
where $\varepsilon>0$ is an arbitrary energy parameter, $T$ is the time over which the time-average is taken, $\|\cdot\|$ denotes the operator norm (specifically the largest singular value of $O$), and $f(\varepsilon,T)$ is given by:
\begin{align}
\label{eq:f_eps_T}    f(\varepsilon,T)=N(\varepsilon)\left(1+\frac{8\log_2 |\sigma(H)|}{\varepsilon T}\right).
\end{align} 
In this function, $N(\varepsilon)$ is the maximum number of energy gaps that $H$ has in any interval of size $\varepsilon$, and $|\sigma(H)|$ is the number of distinct energy eigenvalues of $H$.
The bound~\eqref{eq:ReimannBound} shows that having a large effective dimension is sufficient for $\rho(t)$ to equilibrate relative to $O$ at long times.
Interestingly, equilibration in this context does not necessarily require non-integrability of the quantum system~\cite{Jensen1985,Reimann2013,gogolin2011absence,Yukalov2011}.
In the following, we explore how this observation can serve as the basis for the second law of thermodynamics.

\paragraph*{Entropies relative to observables.}
Given some observable $O$, we can define the Shannon observable entropy $S_\mathrm{Sh}^O$ as in Eq.~\eqref{eq:ShanObsEnt}.
Operationally, this can be understood as the information-theoretic uncertainty or \textit{surprisal} in the measurement outcome of $O$ on the state $\rho$.
For a generic observable, its Shannon entropy can be larger than the von Neumann entropy (e.g.\ for pure states) or smaller than the von Neumann entropy (e.g.\ for a classical mixture of degenerate eigenvectors of the observable).
Another common notion of an entropy with respect to an observable is the observational entropy~\cite{Safranek2019}, which further coarse grains over the eigenspaces of the observable. It is given by 
\begin{align}\label{eq:obs entropy}
    S_\mathrm{Obs}^O[\rho] = -\sum_{i=1}^r p_i \log \frac{p_i}{V_i},
\end{align}
where $V_i=\tr[\Pi_i^O]$ is the dimension of the $i$th eigenspace of the observable $O=\sum_{i=1}^r O_i \Pi_i^O$.
As before, the probabilities $p_i = \tr\left[\Pi_i^O \rho\right]$ are taken with respect to the eigenspaces of $O$.
This entropy has been shown~\cite{Buscemi2023} to be an upper estimate for the von Neumann entropy, satisfying the inequality $S_\mathrm{Obs}^O[\rho] - S_\mathrm{vN}[\rho] \geq D[\rho \| \rho^\mathrm{cg}]\geq0$, where $D[\cdot\|\cdot]$ is the relative entropy and $\rho^\mathrm{cg} = \sum_{i=1}^r \tr\left[\Pi_i^O \rho\right] \frac{\Pi_i^O}{V_i}$ is the coarse-grained state with respect to the observable.
The difference between the Shannon observable entropy and the observational entropy, $S_\mathrm{Obs}^O[\rho] - S_\mathrm{Sh}^O[\rho] = \sum_{i=1}^r p_i \log V_i\geq0$, is commonly referred to as the \textit{Boltzmann term}~\cite{Buscemi2023}.
While the number of eigenspaces $r$ can be understood as the number of macrostates of the observable $O$, the dimension $V_i$ can be understood as the number of microstates, i.e.\ orthogonal states, compatible with a given macrostate $i$.

A priori, it is unclear whether the inclusion of the Boltzmann term in the observational entropy with respect to an observable is necessary to recover a second law, as for example formulated in~\cite{Strasberg2021,Strasberg2024,Nagasawa2024}.
As we show in the following, for $d_\mathrm{eff}\gg r$ where the bound in Eq.~\eqref{eq:shannon_entropy_intro} is tight, the second law holds for both the Shannon observable entropy (Theorem~\ref{thm:shannonbound} in the following) as well as the observational entropy (Proposition~\ref{prop:obsentropybound} in the Appendix).

\section{\label{sec:results}Results}
\paragraph*{The second law on average.}
Starting with the von Neumann entropy, one common and sensible-seeming attempt to formulate a second law for increasing entropy is to look at the entropy of time-averaged states $\langle \rho(t)\rangle_T$, where we have that $S_\mathrm{vN}[\langle \rho(t)\rangle_T] < S_\mathrm{vN}[\omega]$ if we compare to the infinite time averaged state $\omega$ (as we discuss in some more detail in Appendix~\ref{appendix:second_law_vN_entropy}).
While on a formal level, the entropy increases and approaches that of the infinite time average, such an argument fails to capture that in reality, an observer accesses the instantaneous state $\rho(t)$ at a well-defined time $t$, not the time averaged state $\langle \rho(t)\rangle_T$.
It is this time parameter that progresses and with respect to which we observe equilibration, not the duration $T$ of the averaging window.
For these reasons we examine how the Shannon observable entropy of the instantaneous state $\rho(t)$ behaves and whether this quantity satisfies the second law.
By using Jensen's inequality (see Lemma~\ref{lemma:jensenbound} in Appendix~\ref{appendix:proofs_entropy_bounds})
we can arrive at a preliminary statement reminiscent of the second law:
\begin{align}
\label{eq:S_jensen}
    \left\langle S_\mathrm{Sh}^O[\rho(t)]\right\rangle_\infty \leq S_\mathrm{Sh}^O[\omega].
\end{align}
This inequality indicates that the entropy $S_\mathrm{Sh}^O[\rho(t)]$ is on average smaller than the entropy of the equilibrium state $\omega$, aligning with the second law of thermodynamics, which identifies the equilibrium state as having maximal entropy.
Notably, transient fluctuations may result in finite-time states with higher entropy, but the average is bounded above by that of the equilibrium state.

In the following, we expand on this inequality, proving that the entropy is not simply smaller on average than that of the equilibrium state, but that it is closely confined around the equilibrium value.
The key insight needed to reveal the second law on average for equilibrating observables (in the sense of Eq.~\eqref{eq:ReimannBound}) is that it is insufficient to just have the expectation values of the observables equilibrate -- it is also necessary for the probability vector with respect to the observable $O$, $\vec p(t)$, to be close to that of $\vec p_\omega$ for most times.
The distance between the two probability distributions can be characterized by the 1-norm distance $\left\| \vec p(t) - \vec p_\omega \right\|_1 = \frac{1}{2}\sum_{i=1}^r \left|p_i(t) - p_{\omega,i}\right|$. 
Averaging the distance over the time interval $[0,T]$, we find the following bound:
\begin{align}
\label{eq:avg_pt_pomega}
    \left\langle \left\|\vec p(t) - \vec p_\omega \right\|_1\right\rangle_T \leq \eta_{\varepsilon,T}:=\frac{1}{2}\sqrt{\frac{r}{d_\mathrm{eff}}f(\varepsilon,T)},
\end{align}
where $\varepsilon>0$ is an arbitrary energy parameter as in~\eqref{eq:f_eps_T} (see Lemma~\ref{lemma:populationbound} in Appendix~\ref{appendix:proofs_entropy_bounds} for details).
Inequality~\eqref{eq:avg_pt_pomega} improves previous bounds in Refs.~\cite{Short_2012,gogolin2016equilibration} quadratically in terms of their scaling with $r$, showing that in fact $r$ and $d_\mathrm{eff}$ appear with the same exponent.

In contrast to the bound in~\eqref{eq:ReimannBound}, the small parameter $\eta_{\varepsilon,T}$ does not simply vanish in the limit of large effective dimension.
Instead, it is determined by the ratio $r / d_\mathrm{eff}$; for the difference in populations to vanish on average, the effective dimension must be much larger than the number of orthogonal eigenspaces of the observable.
Physically, this is the limit where the number of microstates that participate in the evolution of $\rho(t)$ (captured by $d_\mathrm{eff}$) is much larger than the number of macrostates of the observable (captured by $r$), which is necessary for the population vector $\vec p(t)$ to be close to the equilibrium $\vec p_\omega$ on average.
On a formal level, this also distinguishes our consideration of the equilibration of observables on average from the second law on average, since the former only requires large $d_\mathrm{eff}$, whereas the latter requires large $d_\mathrm{eff}/r$, as we show in the following:

\begin{restatable}[Second law on average]{theorem}{shannonbound}
\label{thm:shannonbound}
Let $\rho(t),$ $\omega$ and $O$ be as defined so far. For arbitrary $\varepsilon>0$ and $T>0$, the following inequality holds:
\begin{align}
\label{eq:S_sh_bound}
    \left\langle \left |S_{\mathrm{Sh}}^O[\rho(t)]-S_{\mathrm{Sh}}^O[\omega]\right|\right\rangle_T \leq \log(r-1)\eta_{\varepsilon,T} + H_2[\eta_{\varepsilon,T}],
\end{align}
so long as $\eta_{\varepsilon,T}<1/2$. The function $H_2[x]:=-x\log x - (1-x)\log(1-x)$ is the binary entropy and $\eta_{\varepsilon,T}$ is defined as in Eq.~\eqref{eq:avg_pt_pomega}.
\end{restatable}

A statement akin to Theorem~\ref{thm:shannonbound} can be derived for the observational entropy, which we discuss in further detail in Appendix~\ref{appendix:generalization_observational_entropy}.
For the observational entropy, it holds that
\begin{align}
    \label{eq:S_obs_bound_main}
    \left\langle \left| S_\mathrm{Obs}^O[\rho(t)] - S_\mathrm{Obs}^O[\omega]\right| \right\rangle_T \leq \log(d)\eta_{\varepsilon,T} +  g\left(\eta_{\varepsilon,T}\right),
\end{align}
where $g(x) = -x\log x + (1+x)\log(1+x)$, and we recall that $d$ is the dimensionality of the Hilbert space. 
The main difference is the prefactor $\log d$ coming from the fact that the observational entropy is bounded above by $\log d$ independently of the observable, while the Shannon observable entropy associated with $O$ is bounded above by $\log r$.

\begin{proof}[Proof sketch.]
The proof consists of two steps:
(i) The first step relates a difference in Shannon entropies to a difference in state populations.
The key here is that if the average distance between $\vec p_\omega$ and $\vec p(t)$ is small, then the average distance between their Shannon entropies $S_\mathrm{Sh}^O[\omega]$ and $S_\mathrm{Sh}^O\left[\rho(t)\right]$ should also be small.
This connection can be made by combining the continuity bound for the Shannon entropy from~\cite{Zhang2007} for the non-averaged quantities together with the result from Eq.~\eqref{eq:avg_pt_pomega}.
(ii) Secondly, we show that under the conditions described above, the population differences are small on average. This is the result from Eq.~\eqref{eq:avg_pt_pomega} and we prove it separately in Lemma~\ref{lemma:populationbound} in Appendix~\ref{appendix:proofs_entropy_bounds}. Both steps together prove the Theorem, and are described in detail in Appendix~\ref{appendix:proofs_entropy_bounds}.
\end{proof}

Note that the theorem, as stated, holds only for ${\eta_{\varepsilon,T}<1/2}$ because the binary entropy is monotonically increasing only for arguments up to $1/2$. If $\eta_{\varepsilon,T}>1/2$, one can simply replace $H_2[\eta_{\varepsilon,T}]$ by $\log 2$, and the bound holds in this modified form (for all values of $\eta_{\varepsilon,T}$). However, in this case the observable does not equilibrate (c.f.\ Eq.~\eqref{eq:avg_pt_pomega}), and thus the bound does not strongly constrain the entropy.

In Fig.~\ref{fig:spinchain_panel} we visualize the behavior described by Theorem~\ref{thm:shannonbound} for a spin chain comprising up to $N=13$ spin-$\frac{1}{2}$ systems interacting under a quantum Ising model with nearest-neighbor $XYZ$ interactions (see Appendix~\ref{sec:details_numerics_model} for details). We work in the parameter regime from~\cite{13KimHuse}, for which the model is nonintegrable, which may not be necessary for equilibration but generally causes faster equilibration times and lower fluctuations~\cite{Jensen1985}.
The dynamics are calculated by exact diagonalization using the Python package QuTiP~\cite{johansson2012qutip}.
The observable in question is the bulk magnetization in the $z$-direction, given by summing all of the individual chain elements' spin-$z$ operators: $M_{z}=\frac{1}{N}\sum_{i=1}^N \sigma_z^{\left(i\right)}$.
For an initial state we choose an uncorrelated chain of spins all pointing downwards $\ket{\downarrow \cdots \downarrow}$.
This has zero entropy with respect to the bulk magentization, thus implementing the past hypothesis.
By combining the coarse-graining-induced information loss associated with accessing the system through a degenerate observable (the global magnetization) together with a low-entropy initial state as in~\eqref{eq:past_hypothesis}, the entropy of the expectation values increases towards its equilibrium value.

\begin{figure*}
    \centering
    \includegraphics[width=\textwidth]{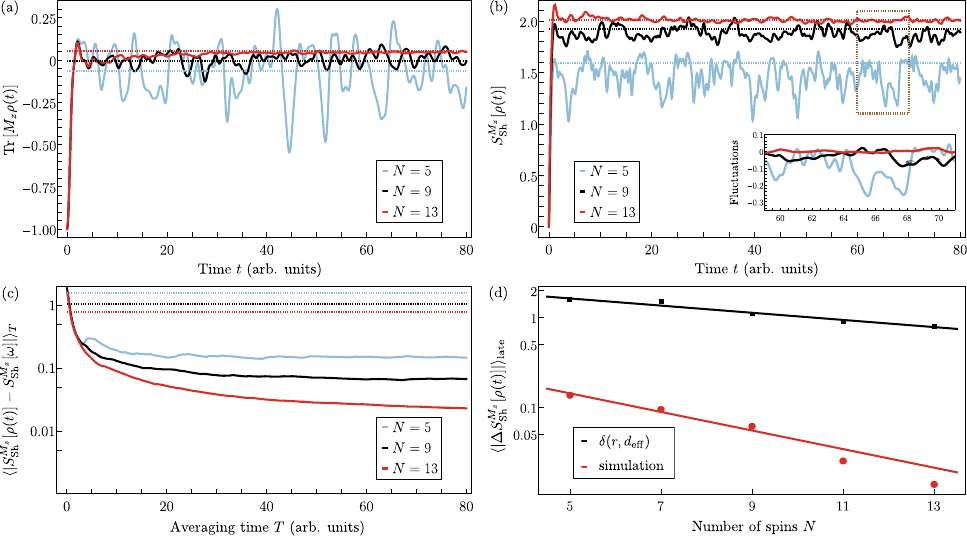}
    \caption{Exact numerical results showcasing the second law on average for isolated spin-chains of 5 to 13 spin-$\frac{1}{2}$ systems with nearest-neighbor interactions. The observable of choice is the global magnetization in the $z$-direction. 
    Panel (a) shows how the magnetization, ${\rm{Tr}}\left[M_z\rho\left(t\right)\right]$ (in units of $\hbar=1$), evolves as a function of time for three different lengths of spin chain, and how it approaches the equilibrium value of $\omega$ (represented by the dashed lines) in each case.
    In panel (b), we plot the Shannon observable entropy as a function of time, finding that entropy increases towards the equilibrium value and then continues to fluctuate around it. Here, the dashed lines represent the Shannon observable entropy of the equilibrium state $\omega$.
    An inset in panel (b) charts the normalized differences $\left(S_{\rm{Sh}}^{M_z}\left[\rho\left(t\right)\right]-S_{\rm{Sh}}^{M_z}\left[\omega\right]\right)/S_{\rm{Sh}}^{M_z}\left[\omega\right]$ at late times.
    Panel (c) shows on a semi-log scale the finite-time averaging from the left-hand side of Eq.~\eqref{eq:S_sh_bound}, as a function of the maximum time $T$ integrated over. The dashed lines show the bounds from the right-hand side of Eq.~\eqref{eq:S_sh_bound} for $T\rightarrow\infty$ (Theorem~\ref{thm:shannonbound}). 
    Lastly, panel (d) plots, on a semi-log scale as a function of the spin chain length $N$, two quantities: the $T\to\infty$ version of the Shannon entropy bound, $\delta\left(r,d_{\rm{eff}}\right)$ (black line and points, see Eq.~\eqref{eq:deltardeff}), and also the averaging of $|\Delta S_{\rm{Sh}}^{M_z}[\rho(t)]|=\left|S_{\rm{Sh}}^{M_z}\left[\rho\left(t\right)\right]-S_{\rm{Sh}}^{M_z}\left[\omega\right]\right|$ over the times $t=50$ to $t=80$ (red line and points).
    For both quantities, an exponential fit of the form $ae^{bx}$ was performed. The fitting paramters plotted here are $a=2.57$, $b=-0.0920$, and $a=0.467$, $b=-0.239$ respectively.} 
    \label{fig:spinchain_panel}
\end{figure*}

It is important to note that the \textit{increase} in entropy stated in the second law is a consequence of the past hypothesis, and not a fundamental feature.
We can see this in Fig.~\ref{fig:spinchain_panel}(b), where there are times where the entropy $S_\mathrm{Sh}^O[\rho(t)]$ is larger than that of the equilibrium $S_\mathrm{Sh}^O[\omega]$. In principle, if we had begun the dynamics at one of these points in time, entropy would decrease towards the equilibrium rather than increase. 
This is clearly illustrated in Fig.~\ref{fig:AppendixFigs2}(b) in the Appendix, where we assume the same system and initial state as in Fig.~\ref{fig:spinchain_panel}, but instead consider the Shannon entropy of the bulk magnetization in the $y$-direction (which does not satisfy the past hypothesis). We find that this entropy actually briefly \textit{decreases} slightly from an initial value close to equilibrium. Note, however, that in this case, the distance of the initial entropy from the equilibrium value is within the range of the late-time fluctuations shown in the figure.

This feature has been extensively discussed throughout the history of statistical mechanics (see e.g.~\cite{lebowitz1993macroscopic} and references therein), and appears naturally in classical settings where equilibration does not necessarily imply that the entropy relative to an observable increases. It is quite possible for the entropy of the equilibrium state to be lower, unless the past hypothesis is invoked.
A simple classical example of this is the separation of oil and water from an initially uniform mixture. If the observable regards the spatial distribution of the oil and water, it seems that the entropy decreases as the two constituents separate. However, the separation results in an increase in the average kinetic energy of the molecules, and a corresponding increase in the entropy associated with that observable.
Thus for one observable the equilibrium state has a lower Shannon entropy than the initial state, and for another observable the situation is reversed. 

\paragraph*{Fluctuation bounds.}
Even when a system has reached equilibrium, the Shannon entropy can fluctuate above and below the equilibrium value for finite times,
with fluctuations above being less likely than fluctuations below, as inequality~\eqref{eq:S_jensen} shows, and also evident in Fig.~\ref{fig:spinchain_panel}(b). Indeed, the Poincar\'{e} recurrence theorem for quantum systems states that any isolated, finite-dimensional quantum system will return arbitrarily often to a state that is arbitrarily close to its initial state~\cite{Bocchieri1957}.
Given that we assume a low-entropy initial state, such a recurrence will result in a reduction in entropy from the equilibrium value by an amount far larger than might be anticipated from the bounded average in~\eqref{eq:S_sh_bound}.
For large systems, however, such apparent violations of the bound occur with vanishing probability at a given time, i.e.\ for systems with large effective dimension compared to number of observable macrostates. This implies that for macroscopic systems, the expected time one would be required to wait to observe such a decrease in entropy occurring is unobservably large.

We now make this statement more quantitative by establishing a fluctuation bound.
Using a probability-theoretic technique similar to~\cite{Romero_2021} it is possible to gain insight into not only the extent of the fluctuations, but also the probability that these fluctuations are larger than some amount.
That is, we can bound the likelihood that the Shannon observable entropy of the state at time $t$ fluctuates away from that of the equilibrium state, depending on the magnitude of the fluctuation.

\begin{restatable}[Shannon entropy fluctuations]{corollary}{shannonfluctuations}
    \label{cor:shannonfluctuations}
    Let $\rho(t),$ $\omega$ and $O$ be as before.
    If $t\in \mathbb R_{\geq 0}$ is sampled uniformly at random, we have
    \begin{align}
    \label{eq:shannon_prob_cheby}
        P\left[\left|S_\mathrm{Sh}^O[\rho(t)] - S_\mathrm{Sh}^O[\omega] \right|\geq \sqrt{\delta}\right]\leq\sqrt{\delta},
    \end{align}
    where $\delta=\delta(r,d_\mathrm{eff})$ is given by
    \begin{align}
    \label{eq:deltardeff}
        \delta(r,d_\mathrm{eff}) = \frac{\log(r-1)}{2}\sqrt{\frac{r}{d_\mathrm{eff}}} + H_2\left[\frac{1}{2}\sqrt{\frac{r}{d_\mathrm{eff}}}\right],
    \end{align}
    with $H_2(\cdot)$ the binary entropy, as in Theorem~\ref{thm:shannonbound}.
\end{restatable}
\begin{proof}[Proof sketch.]
This is a consequence of the generalized Chebyshev inequality $P[ X\geq \kappa] \leq {E[g(X)]}/{g(\kappa)}$,
where $X$ is a positive random variable and $g:\mathbb R\rightarrow \mathbb R$ a monotonously increasing, strictly positive function.
Setting $X=\left|S_\mathrm{sh}[\rho(t)] - S_\mathrm{sh}[\omega] \right|$ and $\kappa:=\sqrt{\delta}$ completes the proof upon applying Eq.~\eqref{eq:S_sh_bound} from Theorem~\ref{thm:shannonbound}. For details, we refer the reader to Appendix~\ref{appendix:proofs_fluctuations}.
\end{proof}

The fluctuation bound in Corollary~\ref{cor:shannonfluctuations} has an interpretation in the sense of typicality~\cite{Goldstein2006,Bartsch2009,reimann2018dynamical,Teufel2023}.
In information theory, when sampling outputs from a fixed source, the weak Law of Large Numbers ensures that for large enough samples, the \textit{likely} samples are those whose entropy is close to that of the source. This fluctuation bound can be interpreted from this perspective, offering a connection between information theory and thermodynamics. In this context, we are sampling at random from a \textit{fixed} source -- the expectation values $\vec{p}$ given by $\tr[O\rho(t)]$ -- and find that at all times it is likely that the entropy of this distribution of expectation values is close to the entropy of the distribution of expectation values of the equilibrium state. As such we may say that the equilibrium state $\omega$ is a `typical' state when sampling expectation values of a state relative to an observable $O$ which equilibrates, and by extension states close to the equilibrium state in $S_\mathrm{Sh}^O[\,\cdot\,]$ are also typical.

This connection emphasizes the ubiquity of information theory for thermodynamics~\cite{rovelli2013relative,Parrondo2015} as it reveals that up to $\sqrt{\delta}$-fluctuations, the measurement of the observable is given by that of the equilibrium average.
Up to a resolution dependent on the extent of fluctuations, and accounting for the finite equilibration time, the description of the state $\rho(t)$ examined through the observable $O$ can therefore be reduced from $d_\mathrm{eff}$ degrees of freedom to the relevant $r$ degrees of freedom of the equilibrium populations $\vec p_\omega$, offering a \textit{compressed} description of the state as a result of typicality.

\section{\label{sec:discussion}Discussion}
We have derived a second law of thermodynamics for isolated quantum systems, clarifying the sense in which its entropy increases over time.
A resolution to the conflict between reversibility at the level of the Schrödinger equation and the irreversibility of the second law comes from considering a measure of entropy defined relative to an observable.
Given that equilibrating observables are highly degenerate, meaning the number of macrostates it defines is much lower than the effective dimension of the system under observation, we have proved novel bounds to show that not only the average value, but in fact the entire probability distribution of observable outcomes tends on average to that of the equilibrium state, despite the overall unitary evolution. As a consequence, the Shannon observable entropy is on average close to the equilibrium value (Theorem~\ref{thm:shannonbound}).
If, furthermore, the system starts in a low entropy state (the past hypothesis), we recover the formulation of the second law that states that the entropy of an isolated system increases over time.

We have also examined how at any point in time in a quantum system's evolution its entropy can grow larger than, or move away from, the entropy of the equilibrium state, as captured in the fluctuation relation given in Corollary~\ref{cor:shannonfluctuations}.
There we bounded the likelihood of fluctuating away from the equilibrium state in terms of the Shannon observable entropy. This points at an interesting connection with information theory. Examining this connection and whether it can be used as a way to compress the description of a quantum system's evolution viewed through a specific observable is left as open question.

\paragraph*{Extension to different entropies.}
Though we discussed entropy defined relative to an observable which can be represented by a self-adjoint operator $O$, one can instead consider a more general positive operator valued measure (POVM) $M=(E_i)_{i=1,\dots,r}$ and denote the corresponding probabilities $p_i(t) = \tr[E_i\rho(t)]$.
In Appendix~\ref{appendix:replacingPOVMs}, we show that Theorem~\ref{thm:shannonbound} and Corollary~\ref{cor:shannonfluctuations} readily generalize to this case without modification.

While this work focused on the Shannon observable entropy, we have proved analogous statements for the observational entropy in Appendix~\ref{appendix:generalization_observational_entropy}, showing that the condition $r/d_\mathrm{eff}\ll1$ is not only sufficient for equilibration of the Shannon observable entropy but also for the observational entropy.
Still, one can construct examples where the Shannon observable entropy fluctuates on the order of $\log r$ and thus does not equilibrate, whereas the observational entropy's fluctuations are relatively small because of the additional Boltzmann term which is of order $\log d$.
This could lead one to conclude that the observable equilibrates while in fact it does not, as the Shannon observable entropy shows.
Regardless of its role in understanding equilibration, the observational entropy does quantify the ignorance resulting from multiple states being compatible with a given measurement outcome~\cite{Safranek2021}, which the Shannon observable entropy alone does not capture. As we have shown, this many-to-one relationship between states and outcomes instead dictates the dynamics of the Shannon entropy, depending on the dynamical participation of those states as encoded in $d_\mathrm{eff}$.

Given that the ratio of observable outcomes to effective dimension in the case of a small subsystem of a much larger total system is generally favorable to equilibration (the former being bounded by the dimensionality of the subsystem), our results can likely be used to constrain the behavior of the von Neumann entropy of a subsystem. Since our focus was on the overall behavior of isolated systems, we have not addressed this here. 
Additionally, one may ask how our analysis extends to other notions of entropy.
One such example is the diagonal entropy defined relative to the energy eigenbasis as in~\cite{Polkovnikov2011,Ikeda2015}, which is trivially constant for an isolated quantum system.
When defining a diagonal entropy in the context of observables akin to~\cite{Scarpa2023}, however, a non-trivial behavior is recovered (details in Appendix~\ref{appendix:diagonal_entropy}).
Other examples include the family of Tsallis or Rényi entropies~\cite{Tsallis1988,Hu2006,Vershynina2019} defined relative to an observable in analogy to how we define the Shannon observable entropy.
By combining entropy continuity bounds such as those of e.g.~\cite{Hu2006,Rastegin2013} together with Lemma~\ref{lemma:populationbound} for PVMs or Lemma~\ref{lemma:POVMpopbound} for POVMs (see Appendix), respective formulation of the second law might be obtained for such entropies.

\paragraph*{On the classical second law formulations.}
A natural question following from the results presented here is that of how our statement of the second law compares to the historical formulations of Clausius~\cite{Clausius1854} (``No process is possible, the sole result of which is the transfer of heat from a cold body to a hot body'') or Kelvin-Planck~\cite{Planck1897,Thess2011}  (``No process is possible, the sole result of which is that a body is cooled and work is done'').
A rigorous analysis of these statements would require a definition of work, heat and temperature for subsystems of isolated quantum systems, which would be relative to the observable the agent can use to examine and manipulate the system.
In certain special cases of equilibrating systems, the classical thermodynamic notions of entropy can be recovered -- in which case we have \textit{thermalization}.
One paradigmatic setting for thermalization comprises a subsystem $S$ weakly coupled to a larger environment $E$.
Under suitable assumptions as exemplarily outlined in the seminal works~\cite{Popescu2006,Goldstein2006,linden2009quantum}, the equilibrium state $\omega$ on the global system is typically well approximated by the thermal state $e^{-\beta H_S}/Z_S$ of $S$.
Here, $\beta$ is some temperature parameter associated with the energy and $Z_S$ the canonical partition function ensuring normalization of the thermal state.
When the observable under consideration is the Hamiltonian $O=H_S\otimes\mathds 1_E$ of the subsystem, the equilibrium Shannon observable entropy is given by the Gibbs thermal entropy $S_{\rm Obs}^{H_S}[\omega] \approx \log Z_S + \beta\langle H_S\rangle$ with Boltzmann constant $k_B=1$, up to deviations coming from the difference between the actual equilibrium state and the Gibbs state, as we outline in more detail in Appendix~\ref{appendix:thermodynamic_entropy}.
Theorem~\ref{thm:shannonbound} then states that as the subsystem's state evolves, its Shannon observable entropy (relative to the subsystem Hamiltonian) also converges to the thermal equilibrium value, in agreement with classical formulations of the second law.
How a Gibbs-like form of the entropy relative to an observable can also be recovered in a more general context beyond the system-environment paradigm was investigated in~\cite{Scarpa2023}, finding that such a form can encode equilibrium behaviour in various examples.

Finally, our work corroborates the findings of previous studies noting that non-integrability is not a necessary requirement for equilibration relative to observables~\cite{Jensen1985,Yukalov2011,gogolin2011absence,Reimann2013}.
As we showed in~\eqref{eq:avg_pt_pomega}, a large $d_\mathrm{eff}/r$ suffices for equilibration of the observable probability to the equilibrium distribution, and thus for our formulation of the second law to hold, and this does not a priori require non-integrability.
Whether the system is integrable or not may still affect some details like the speed of equilibration or tightness of the bounds as noted in~\cite{Jensen1985,gogolin2016equilibration}.
In other contexts as well, a connection between integrability and effective dimension has been hinted at~\cite{Iyoda2018,Prazeres2023}, which leads us to pose the question of whether $d_\mathrm{eff}$ could allow for a quantitative definition of integrability in relation to the equilibration of quantum systems.

\acknowledgments
The authors thank Pharnam Bakshinezhad, Felix Binder and Sandu Popescu for valuable discussions,
Joe Schindler for comments on the generalization of our bounds to POVMs,
and Dominik Safranek for remarks on the fluctuation bounds.
T.R.\ thanks Sophie Engineer for assistance with programming.
The authors acknowledge the TU Wien Bibliothek for financial support through its Open Access Funding Programme.
T.R., J.X.\ and M.H.\ would like to acknowledge funding from the European Research Council (Consolidator grant `Cocoquest’ 101043705). 
T.D.\ acknowledges support from {\"O}AW-JESH-Programme and the Brazilian agencies CNPq (Grant No.\ 441774/2023-7 and 200013/2024-6) as well as INCT-IQ through the project (465469/2014-0). 
This project is co-funded by the European Union (Quantum Flagship project ASPECTS, Grant Agreement No.\ 101080167). 
Views and opinions expressed are however those of the authors only and do not necessarily reflect those of the European Union, REA or UKRI. Neither the European Union nor UKRI can be held responsible for them. 
This publication was made possible through the support of Grant 62423 from the John Templeton Foundation. The opinions expressed in this publication are those of the author(s) and do not necessarily reflect the views of the John Templeton Foundation.

\hypertarget{sec:appendix}
\appendix


\section*{Appendices}

\section{\label{appendix:proofs_entropy_bounds}Proofs for entropy bounds}
In this section, we prove all our main results and provide additional background to the statements and the formalism we use.
In Sec.~\ref{appendix:proofs_average} we provide the proofs for the average entropy difference bounds, that is Theorem~\ref{thm:shannonbound} and Proposition~\ref{prop:obsentropybound}. Following this, Sec.~\ref{appendix:proofs_fluctuations} contains the proofs for the fluctuation bounds of Corollary~\ref{cor:shannonfluctuations} and Corollary~\ref{cor:observablefluctuations}.

\subsection{\label{appendix:proofs_average}Proofs for averaged entropy differences}
Before discussion of the proof of Theorem~\ref{thm:shannonbound}, more details regarding the inequality~\eqref{eq:S_jensen} are supplied, in particular regarding the applicability of Jensen's inequality:
\begin{restatable}{lemma}{jensenbound}
\label{lemma:jensenbound}
    Given $\rho(t)$, $\omega$ and $O$ as defined in the main text, we have in the infinite time-average,
    \begin{align}
    \label{eq:S_jensen_SM}
        \left\langle S_\mathrm{Sh}^O[\rho(t)]\right\rangle_\infty \leq S_\mathrm{Sh}^O[\omega].
    \end{align}
\end{restatable}
\begin{proof}
The proof of this statement is an application of the Jensen inequality which states that for a concave, real-valued function $g[X]$ the following inequality holds: $\left\langle g[X]\right\rangle \leq g\left[\left\langle X \right\rangle\right]$.
In this context, $X$ is a real-valued random variable that can in general be defined on a convex subset of $\mathbb R^r$, as shown for example in Theorem~7.11 of~\cite{Klenke2020}.
For our setting, we have that $g[\,\cdot\,]=S_\mathrm{Sh}^O[\,\cdot\,]$ is concave as a function of the probability vector $\vec p(t)$ defined by the observable outcomes.
Note that it is not concave as a function of $t$, though this is, of course, not needed here.
Formally, the Shannon observable entropy as defined in the main text is a function of density matrices, but since it only depends on the population vectors, we simplify the notation here to simply write $S_\mathrm{Sh}^O[\vec p]$ to actually mean $S_\mathrm{Sh}^O[\rho]$, where $\vec p$ is the population vector of $\rho$ relative to the observable $O$.
By identifying $X=\vec p(t)$, and taking the uniform measure over the space $t\in [0,T]$, Jensen's inequality reads
\begin{align}
    \frac{1}{T}\int_0^T dt  S_\mathrm{Sh}^O[\vec p(t)] \leq S_\mathrm{Sh}^O\left[\frac{1}{T}\int_0^T dt\vec p(t)\right].
\end{align}
Upon taking the limit $T\rightarrow\infty$, this yields for the left-hand side $\langle S_\mathrm{Sh}^O[\rho(t)]\rangle_\infty$ and for the right-hand side (by using continuity of the Shannon entropy), we can pull the limit $t\rightarrow\infty$ inside the entropy to find 
\begin{align}
    \lim_{t\rightarrow\infty} S_\mathrm{Sh}^O\left[\frac{1}{T}\int_0^T dt\vec p(t)\right] &= S_\mathrm{Sh}^O[\vec p_\omega].
\end{align}
Furthermore, we have used linearity of the trace, $\langle p_i(t)\rangle_\infty = \langle \tr\left[\Pi_i^O \rho(t) \right]\rangle_\infty = \tr\left[\Pi_i^O \langle \rho(t) \rangle_\infty\right] = p_{i,\omega}$ to conclude that the infinite time average of the population equals the population of the infinite time average, completing the proof of the Lemma.
\end{proof}

Moving on towards proving our main result, Theorem~\ref{thm:shannonbound}, we devise the following proof strategy:
\begin{itemize}
    \item First, we want to determine conditions under which the populations of $\rho(t)$ and $\omega$ with respect to the observable $O$ do not differ much on average.
    \item Second, once we know that the populations of the two states are close to each other for most times, we can show that the Shannon observational entropy (as well as the observational entropy) of $\rho(t)$ is also close to that of $\omega$ for most times.
\end{itemize}
We start with the first part. Let us recall that $\vec p(t) = (p_1,\dots,p_r)$ with $p_i=\tr[\Pi_i^O \rho]$ is the vector of populations of $\rho(t)$ with respect to the observable's eigenspaces.
Similarly, $\vec p_\omega$ is defined for the infinite time averaged state $\omega$.
Then, in Eq.~\eqref{eq:avg_pt_pomega} of the main text we claimed that $\langle\|\vec p(t) - \vec p_\omega\|_1\rangle_T$ is small, which we prove in detail as part of Lemma~\ref{lemma:populationbound}.
The key reason why the populations' differences are small on average is based on the results for observables from the theory of equilibration on average, which state that $\tr[O\rho(t)]$ is close to $\tr[O\omega]$ most of the time.
One subtle point here concerns the number of macrostates: while for the observables' expectation values it was sufficient to have that $d_\mathrm{eff}$ is large for equilibration, now, the effective dimension has to be large compared to $r$, the number of macrostates corresponding to the observable $O$.
We can therefore state that observable expectation values equilibrate under weaker conditions than the population vectors.

\begin{restatable}[Equilibrating populations]{lemma}{populationbound}
\label{lemma:populationbound}
    Let $\vec p(t)$ and $\vec p_\omega$ be as in the introduction preceding this Lemma.
    Furthermore, take $\varepsilon>0$ and $T>0$ to be arbitrary, then,
    \begin{align}
        \left\langle \left\|\vec p(t) - \vec p_\omega \right\|_1\right\rangle_T \leq \frac{1}{2}\sqrt{\frac{r}{d_\mathrm{eff}}f(\varepsilon,T)} =: \eta_{\varepsilon,T},
    \end{align}
    where the constant $\eta_{\varepsilon,T}$ is defined with $f(\varepsilon,T)$ like in Eq.~\eqref{eq:f_eps_T}.
\end{restatable}

\begin{proof}
Note that since our result scales with $\sqrt{r}$, it is in fact a tighter bound than the one shown in Theorem~2 of Ref.~\cite{Short_2012}, where the upper bound scales with $r$.
To start the proof, let us recall the definition of $\|\cdot \|_1$ here, which is given by
\begin{align}
    \|\vec p(t) - \vec p_\omega \|_1 = \frac{1}{2}\sum_{i=1}^r \left|\tr\left[\Pi_i^O \rho(t)\right]-\tr\left[\Pi_i^O \omega\right]\right|.
    \label{eq:p_1norm}
\end{align}
We are interested in an upper bound on the time-average of this quantity.
To use the bounds on observable equilibration on average~\cite{reimann2008foundation,Short_2012,gogolin2016equilibration}, we need to consider time-averages of squares of trace differences and not only trace differences.
To be explicit, we want to use the intermediate Eq.~(9) from~\cite{Short_2012}, which states
\begin{align}
    &\left\langle \left|\tr\left[\Pi_i^O \rho(t)\right]-\tr\left[\Pi_i^O \omega\right]\right|^2\right\rangle_T \nonumber \\
    &\quad \leq f(\varepsilon,T) \sqrt{\tr\left[\Pi_i^{O\dagger} \Pi_i^O \omega^2\right]\tr\left[\Pi_i^O \Pi_i^{O\dagger} \omega^2\right]}.
    \label{eq:p_2norm}
\end{align}
Equation~\eqref{eq:p_1norm} can be converted into a suitable form to use Eq.~\eqref{eq:p_2norm} by making the following transformations,
\begin{align}
    \left\langle \left\|\vec p(t)- \vec p_\omega \right\|_1\right\rangle_T &\leq \frac{\sqrt{r}}{2}\left\langle \left\|\vec p(t) -\vec p_\omega \right\|_2\right\rangle_T \label{eq:12normineq}\\
    &\leq \frac{\sqrt{r}}{2}\sqrt{\left\langle \left\|\vec p(t)- \vec p_\omega \right\|_2^{\, 2}\right\rangle_T}. \label{eq:jensenSQRT}
\end{align}
In the first step for Ineq.~\eqref{eq:12normineq} we made use of the norm-inequality $2\|\cdot\|_1 \leq {\sqrt{r}}\|\cdot \|_2$ to switch from the $1$-norm to the $2$-norm, and in the second step for Ineq.~\eqref{eq:jensenSQRT} we apply the Jensen inequality $\langle \sqrt{X}\rangle \leq \sqrt{\langle X\rangle}$ using concavity of the square root function $\sqrt{\cdot}$.
Now, we can use Eq.~\eqref{eq:p_2norm} together with the fact that orthogonal projectors are Hermitian, $\Pi_i^O = \Pi_i^{O\dagger}$, and satisfy $(\Pi_i^O)^2=\Pi_i^O$, to write
\begin{align}
    \left\langle\left\|\vec p(t) -\vec p_\omega\right\|_2^{\,2}\right\rangle_T &\leq f(\varepsilon,T)\sum_{i=1}^r \tr\left[\Pi_i^O \omega^2\right] \label{eq:Pi_omega2} \\
    & = f(\varepsilon,T) \tr\left[\omega^2\right] \label{eq:Pi_iO_completeness} \\
    &= \frac{f(\varepsilon,T)}{d_\mathrm{eff}}.\label{eq:omega2_deff}
\end{align}
For Eq.~\eqref{eq:Pi_iO_completeness} we have used the completeness relation $\sum_{i=1}^r \Pi_i^O = \mathds 1$ and for Eq.~\eqref{eq:omega2_deff} we have used the definition of the effective dimension $d_\mathrm{eff}$ also stated in the main text.
Combining this together with Eq.~\eqref{eq:jensenSQRT} yields the desired result.
\end{proof}

Now that we have checked the first item of our proof strategy outlined in the beginning of Sec.~\ref{appendix:proofs_entropy_bounds}, we are ready to show Theorem~\ref{thm:shannonbound}, which we restate in the following for completeness.

\shannonbound*

\begin{proof}
We subdivide the proof of this statement into two steps: (1) We express the left-hand side of Eq.~\eqref{eq:S_sh_bound} in terms of population differences, giving us a bound on the entropy difference on average expressed as a function of the average difference in populations between $\rho(t)$ and the infinite-time averaged state $\omega$. (2) Then, we use the result from Lemma~\ref{lemma:populationbound}, which we can then insert into the bound found in the first step, finalizing the proof of this Theorem.

\paragraph*{Step (1).} We start by using the continuity bound for the Shannon entropy derived in \cite{Zhang2007} which states that
\begin{align}
\label{eq:entropy_continuity_bound}
    |S_\mathrm{sh}[\vec p] - S_\mathrm{sh}[\vec q]| \leq \log(r-1)\|\vec \vartheta\|_1 + H\left[\|\vec \vartheta\|_1\right],
\end{align}
for all $r$-dimensional probability vectors $\vec p$ and $\vec q,$ with $\vec\vartheta = \vec p - \vec q,$ and $S_\mathrm{Sh}[\vec p] = -\sum_{i=1}^r p_i\log p_i$ the Shannon entropy. Here, the 1-norm is defined as in the main text,
\begin{align}
\label{eq:1-norm_vectors}
    \|\vec \vartheta\|_1 = \frac{1}{2}\sum_{i=1}^r |\vartheta_i|,
\end{align}
where we highlight the factor $\frac{1}{2}$, and $H$ is again the binary entropy.
By taking $\vec p(t)$ as the first vector and $\vec q := \vec p_\omega$ as the second one, we may now apply the continuity bound from Eq.~\eqref{eq:entropy_continuity_bound} to the Shannon entropies from the Lemma, with $\vartheta(t) = \vec p(t) - \vec p_\omega.$
Averaging preserves inequalities, and consequentially, we find that
\begin{align}
    &\left\langle \left |S_{\mathrm{Sh}}^O[\rho(t)]-S_{\mathrm{Sh}}^O[\omega]\right|\right\rangle_T \nonumber \\
    &\qquad = \left\langle \left |S_{\mathrm{Sh}}[\vec p(t)]-S_{\mathrm{Sh}}[\vec p_\omega]\right|\right\rangle_T \\
    &\qquad \leq \log(r-1)\langle\|\vec \vartheta\|_1 \rangle_T + \langle H_2[\|\vec \vartheta\|_1]\rangle_T.
\end{align}
Now, we use that the binary entropy $H$ is concave, which allows us to apply Jensen's inequality once again,
\begin{align}
    \langle H_2[\|\vec \vartheta\|_1]\rangle_T \leq H_2[\langle \|\vec \vartheta\|_1\rangle_T],
\end{align}
finding an upper bound of the average binary entropy.
Thereby, we conclude the first step which relates the average difference of Shannon entropies to an average of the differences in the populations,
\begin{align}
    &\left\langle \left |S_{\mathrm{Sh}}^O[\rho(t)]-S_{\mathrm{Sh}}^O[\omega]\right|\right\rangle_T \nonumber \\
    &\qquad\leq \log(r-1)\langle\|\vec \vartheta\|_1 \rangle_T + H_2[\langle \|\vec \vartheta\|_1\rangle_T].
    \label{eq:S_diff_step1}
\end{align}

\paragraph*{Step (2).}
Recalling that the binary entropy $H_2[x]$ is monotonously growing for $x\in[0,1/2]$ allows us to insert the inequality from Lemma~\ref{lemma:populationbound}, $\langle\|\vec\vartheta\|_1\rangle_T \leq \eta_{\varepsilon,T}$ into the result from the first step~\eqref{eq:S_diff_step1}. We find,
\begin{align}
    \left\langle \left |S_{\mathrm{Sh}}^O[\rho(t)]-S_{\mathrm{Sh}}^O[\omega]\right|\right\rangle_T \leq \log(r-1)\eta_{\varepsilon,T} + H\left[\eta_{\varepsilon,T}\right].
\end{align}
One subtle point is that the binary entropy grows monotonously only for $x\in[0,1/2]$, hence, the result only holds for values $0<\eta_{\varepsilon,T} \leq 1/2$.
This is not restrictive for two reasons: first, by replacing $H(x)$ with $\log 2$ for arguments greater than $1/2$ the result holds for all values of $\eta_{\varepsilon,T}$ and secondly, if $\eta_{\varepsilon,T}>1/2$, we are anyways in a regime where the bound is not tight -- i.e.\ $\log r$ is the maximum Shannon observable entropy attainable, and in the regime of relatively large $\eta_{\varepsilon,T}>1/2$, the difference between $S_\mathrm{Sh}^O[\rho(t)]$ and $S_\mathrm{Sh}^O[\omega]$ is of the order of the entropies themselves.
Since this technicality does not affect the qualitative result in the relevant limit of $d_\mathrm{eff} \gg r$ of equilibration, we have forgone detailing it in the main text.
This concludes the proof of the Theorem.
\end{proof}

\begin{figure*}
    \centering
    \includegraphics[width=\textwidth]{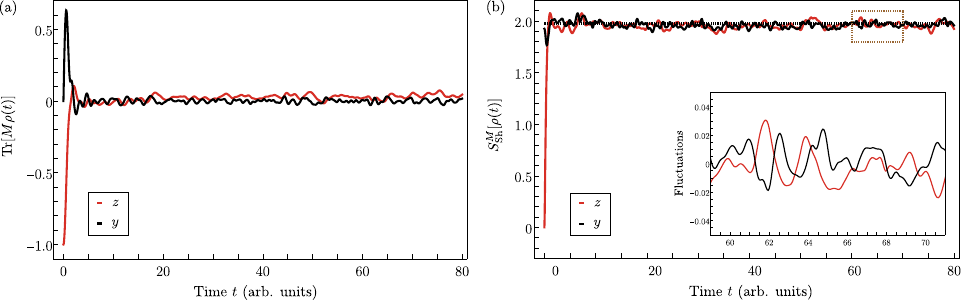}
    \caption{For two observables, the bulk magnetization $M_z$ in the $z$-direction and the bulk magnetization $M_y$ in the $y$-direction, the expectation values (panel (a)) and the Shannon entropy (panel (b)) are plotted as functions of time, for a spin chain of length 11, for the same initial state and Hamiltonian as in the main text.
    The inset in panel (b) shows the relative fluctuations at late times, as the equivalent inset in Fig.~\ref{fig:spinchain_panel}(b) does.
    We see that being initially in a low-entropy state with respect to the $z$-magnetization here means to be in a high-entropy state with respect to the $y$-magnetization, and in the transient regime the $y$-magnetization first fluctuates strongly compared to the late time fluctuations (see panel (a)) and as a result, the Shannon entropy also initially fluctuates to a lower-than-equilibrium value, as seen for short times in panel (b).}
    \label{fig:AppendixFigs2}
\end{figure*}

\subsection{\label{appendix:proofs_fluctuations}Proofs for fluctuation theorems}

Here we provide a proof for the fluctuation bounds for the Shannon observable entropy (Corollary~\ref{cor:shannonfluctuations}), which can also by straightforward analogy serve as the proof of Corollary~\ref{cor:observablefluctuations} concerning the fluctuations of the observational entropy.

\shannonfluctuations*

\begin{proof}
    We provide additional details for the proof of Corollary~\ref{cor:shannonfluctuations} by showing a simple proof for the generalized Chebyshev inequality, which we recall here for completeness:
    \begin{align}
        P[ X\geq \kappa] \leq \frac{E[g(X)]}{g(\kappa)}.
    \end{align}
    We can re-write the probability as an integral over the corresponding measure-space $\Omega$ with measure $\mu$,
    \begin{align}
    \label{eq:generalized_chebyshev}
        P[ X\geq \kappa] = \int_\Omega \mathds 1_{\{X\geq \kappa\}} d\mu,
    \end{align}
    with $\mathds 1_{\{X\geq \kappa\}}$ being the characteristic function of the pre-image of ${\{X\geq \kappa\}}\subseteq \Omega$.
    Because $g$ is monotonously increasing and $g(\kappa)>0$ by assumption, we have $\mathds 1_{\{X\geq \kappa\}} \leq g(X)/g(\kappa)$ everywhere. From this, it follows that
    \begin{align}
        \int_\Omega \mathds 1_{\{X\geq \kappa\}} d\mu\leq \frac{1}{g(\kappa)}\int_\Omega g(X) d\mu = \frac{E[g(X)]}{g(\kappa)},
    \end{align}
    which is the Chebyshev inequality as claimed.
    For our purposes, the measure space is $\Omega =\mathbb R_{\geq 0}$ and the random variable $X$ is given by the mapping
    \begin{align}
        t\mapsto X:= \left|S_\mathrm{Sh}^O[\rho(t)] - S_\mathrm{Sh}^O[\omega] \right|. 
    \end{align}
    The measure $\mu$ can be taken to be a uniform distribution over the interval $[0,T]$ for which we can take the limit $T\rightarrow\infty$. This results in the expectation value
    \begin{align}
        E[X] = \frac{{\log(r-1)}}{2}\sqrt{\frac{r}{d_\mathrm{eff}}} + H_2\left[\frac{1}{2}\sqrt{\frac{r}{d_\mathrm{eff}}}\right] =: \delta,
    \end{align}
    as a direct consequence of Theorem~\ref{thm:shannonbound}. Setting $\kappa:=\sqrt{\delta}$ yields
    \begin{align}
        P[X\geq \sqrt{\delta}] \leq \frac{\delta}{\sqrt{\delta}}=\sqrt{\delta},
    \end{align}
    which completes the proof as claimed in the main text.
\end{proof}

Instead of using the symmetric fluctuation bounds in Corollary~\ref{cor:shannonfluctuations} where both the magnitude of the fluctuations as well as the probability of them is compared to $\sqrt{\delta}$, one can use a more general re-scaled form.
Starting with the Chebyshev inequality~\eqref{eq:generalized_chebyshev} we could instead write a more general bound
\begin{align}
\label{eq:generalized_fluctuations_shannon}
        P\left[\left|S_\mathrm{Sh}^O[\rho(t)] - S_\mathrm{Sh}^O[\omega] \right|\geq \mu \right]\leq \frac{\delta}{\mu},
\end{align}
where on the left-hand side we consider fluctuations of the Shannon entropy of the order of $\mu$.
At the same time, on the right-hand side, the inverse is multiplied giving $\delta/\mu$.
This generalized bound may be useful in the setting where one considers entropy fluctuations of a specific magnitude $\mu>\sqrt{\delta}$ and therefore the bound in Corollary~\ref{cor:shannonfluctuations} does not directly apply anymore.
The Chebyshev inequality is in general not tight, though, and therefore it is expected that for some choices of parameters, the probability of fluctuations is overestimated by this method.

\section{Additional entropies}
As we discussed in the main text, there is no unique entropy with respect to which one can examine irreversibility in the context of isolated quantum system's evolution, and given different measures of entropy than the ones we used, one would have to adapt the results.
For the main part of this work, we have focused on the Shannon entropy of an observable.
The von Neumann entropy, however, is arguably the most commonly used entropy measure for quantum states, and thus, we discuss in Sec.~\ref{appendix:second_law_vN_entropy} how to understand the von Neumann entropy in the context of equilibration.
In the context of observables the observational entropy is another common notion of entropy, and in Sec.~\ref{appendix:generalization_observational_entropy} we provide the results from our main text adapted to the observational entropy.

\subsection{\label{appendix:second_law_vN_entropy}A second law for the von Neumann entropy?}
Give some state $\rho,$ its von Neumann entropy is constant under arbitrary unitary transformations $S_\mathrm{vN}[\rho]=S_\mathrm{vN}[U\rho U^\dagger]$~\cite{Nielsen2010}.
As a consequence, as already stated in the main text, if a state $\rho(t)$ evolves unitarily according to some Hamiltonian $H$, we have that $S_\text{vN}[\rho(0)]=S_\text{vN}[\rho(t)]$.
Going towards time-averaged states,
\begin{align}
\label{eq:timeavg}
    \langle\rho(t)\rangle_T = \frac{1}{T}\int_0^T dt \rho(t),
\end{align}
with $\omega = \lim_{T\rightarrow\infty}\langle \rho(t)\rangle_T,$ it is given that $S_\mathrm{vN}[\rho(t)]<S_\mathrm{vN}[\omega].$
Looking at the entropy $S_\mathrm{vN}[\langle\rho(t)\rangle_T]$ of the time-averaged state,
this function does in general not monotonously increase with $T$.
A precessing spin $\ket{\psi(t)}=\sin(gt)\ket{\uparrow} + \cos(gt)\ket{\downarrow}$ already provides a counterexample where the entropy of the time-averaged state is not monotonously growing in $T$.
Due to the continuity of the von Neumann entropy~\cite{Fannes1973}, the entropy of the finite-time averaged state tends to that of the infinite-time averaged state as we increase the averaging time $T$, i.e., 
$S_\text{vN}[\langle\rho(t)\rangle_{T}] \to S_\text{vN}[\omega]$ as $T\to\infty$.
Quantitatively, the convergence of $S_\mathrm{vN}[\langle\rho(t)\rangle_T]$ to $S_\mathrm{vN}[\omega]$ is captured by the following result:

\begin{restatable}[Von Neumann entropy of time-averaged states]{proposition}{vNtimeAvg}
\label{prop:vNtimeAvg}
    Let $\rho(t)$ be the time-evolution of some initial state $\rho_0$ according to the Hamiltonian $H$ with spectrum $\{\lambda_i\}_i$, all defined on a $d$-dimensional Hilbert space, and let $\langle \rho(t)\rangle_T$ be the time-averaged state as in Eq.~\eqref{eq:timeavg}.
    Then,
    \begin{align}
    \label{eq:S_vN_diff}
        \left|S_\mathrm{vN}\left[\langle\rho(t)\rangle_T\right] - S_\mathrm{vN}\left[\omega\right]\right| \leq \log(d) \vartheta + H_2[\vartheta],
    \end{align}
    with $\vartheta = {2\sqrt{d}}/{(\omega_\mathrm{min}T)}$ and $\omega_\mathrm{min} := \min\{|\lambda_i-\lambda_j| \neq 0\}$ is the smallest non-zero magnitude gap of $H$.
\end{restatable}
\begin{proof}
The statement follows from using the tight Fannes-Audeneart continuity bound for the von Neumann entropy~\cite{Fannes1973,Audenaert2007} which states
\begin{align}
\label{eq:S_vN_FannesAudeneart}
    &\left|S_\mathrm{vN}\left[\langle\rho(t)\rangle_T\right] - S_\mathrm{vN}\left[\omega\right]\right| \\
    &\quad \leq \frac{1}{2}\log(d) \|\langle\rho(t)\rangle_T - \omega\|_1 + H_2\left[\frac{1}{2}\|\langle\rho(t)\rangle_T - \omega\|_1\right]. \nonumber
\end{align}
To determine the desired upper bound in Eq.~\eqref{eq:S_vN_diff} we must find how quickly the $1$-norm distance between $\langle\rho(t)\rangle_T$ and $\omega$ converges to zero.
The $1$-norm is given by
\begin{align}
\label{eq:1-norm-quantum}
    \left\|\sigma\right\|_1 := \tr\left[\sqrt{\sigma^\dagger\sigma}\right],
\end{align}
and can be upper bounded by the $2$-norm using the Cauchy-Schwartz inequality~\cite{Coles2019},
\begin{align}
    \left\|\sigma \right\|_1 &= \tr\left[\mathds 1 \sqrt{\sigma^\dagger\sigma}\right] \\
    &\leq \sqrt{\tr\left[\mathds 1^\dagger \mathds 1\right]\tr\left[\sqrt{\sigma^\dagger\sigma}^\dagger \sqrt{\sigma^\dagger\sigma}\right]} \\
    &=\sqrt{d}\left\|\sigma\right\|_2.\label{eq:1_2_norm_equiv}
\end{align}
With the $2$-norm at hand, we can make analytical progress.
To that end, let us consider $\rho(t)$ written in the energy-eigenbasis decomposition   with respect to the Hamiltonian $H$ generating the evolution, which allows us to write
\begin{align}
    \left\langle\rho(t)\right\rangle_T &= \sum_{ij} \rho_{ij}\left\langle e^{-i\omega_{ij}t} \right\rangle_T \ketbra{i}{j},
\end{align}
where $\ket{i}$ is an energy eigenstate indexed by the label $i$ and $\omega_{ij}=\lambda_i-\lambda_j$ is the difference of energies between $i$th and $j$th energy eigenstate. Note that for degenerate Hamiltonians there is the possibility that $\omega_{ij}=0$ for $i\neq j.$
We can write
\begin{align}
    \left\langle e^{-i\omega_{ij}t} \right\rangle_T &= e^{-i\omega_{ij}T/2}\mathrm{sinc}\left(\frac{\omega_{ij}T}{2}\right),
\end{align}
with $\mathrm{sinc}(x) = \sin(x)/x,$ with the continuous continuation $\mathrm{sinc}(0)\equiv 1$ to capture the limit $T\rightarrow 0$ or the case $\omega_{ij}=0.$
Then, the $2$-norm distance between the two states can be written exactly as
\begin{align}
    \left\|\left\langle\rho(t)\right\rangle_T-\omega\right\|_2^{\, 2} &=\tr\left[\left(\left\langle\rho(t)\right\rangle_T-\omega\right)^\dagger \left(\left\langle\rho(t)\right\rangle_T-\omega\right)\right] \nonumber \\
    &=\sum_{ij} \left|\rho_{ij}\right|^2 \left|\mathrm{sinc}\left(\frac{\omega_{ij}T}{2}\right) - \delta_{\omega_{ij}0}\right|^2 \nonumber \\
    &= \sum_{ij\,:\,\omega_{ij}\neq 0} \left|\rho_{ij}\right|^2 \left|\mathrm{sinc}\left(\frac{\omega_{ij}T}{2}\right)\right|^2 \nonumber \\
    &\leq \left(\frac{2}{\omega_\mathrm{min} T}\right)^2,\label{eq:2_norm_res}
\end{align}
where we have set $\omega_\mathrm{min}=\min \{|\omega_{ij}|\neq 0\}$ to be the energy gap of $H$ which is of smallest non-zero magnitude.
Combining the results from Eq.~\eqref{eq:1_2_norm_equiv} with Eq.~\eqref{eq:2_norm_res}, we find
\begin{align}
    \|\langle\rho(t)\rangle_T - \omega\|_1 &\leq \frac{2\sqrt d}{\omega_\mathrm{min}T},
\end{align}
which we can directly insert into the entropy inequality~\eqref{eq:S_vN_FannesAudeneart}, to find the upper bound with $\vartheta = \frac{2\sqrt{d}}{\omega_\mathrm{min}T}$, as stated in Proposition~\ref{prop:vNtimeAvg}.
Note that the statement in the given form only holds for values $\vartheta\leq 1/2$, again because $H_2$ is monotonously growing only for arguments $\leq 1/2$.
To remedy this, one could replace $H_2$ by $\log 2$ for arguments $\vartheta > 1/2$, to arrive at a more general statement.
\end{proof}

A comment is in order at this point regarding the convergence rate of the von Neumann entropy of the time-averaged state to that of the equilibrium state.
The bound states that in general $T\geq\sqrt{d}/\omega_\mathrm{min}$ is required for the entropies to be reasonably close; for macroscopically large systems where $d$ is exponential in the system size, such a convergence is of course not of practical relevance. Furthermore, this does not match the observation that equilibration, in practice, occurs in finite time.
The origin of the factor $\sqrt{d}$ comes from the norm-equivalence $\|\cdot\|_1\leq \sqrt{d}\|\cdot\|_2,$ and this raises the question of how tight the given bound is.
When arriving at the final inequality~\eqref{eq:p_2norm}, the inequality we used is in general not very tight, i.e.\ we have simply taken the slowest of all the time-scales to appear in the bound, whereas in practice a tighter bound can be found by taking a weighted average over the time-scales.
There, in practice, $\|\langle\rho(t)\rangle_T-\omega\|_2$ may scale as $1/\sqrt{d}$, if the Hamiltonian is sufficiently gapped and the initial state has support on the entire spectrum (like, for, example an equal superposition over all energy eigenstates of a capped harmonic oscillator). In this case the actual bound will in many instances be tighter than the one stated in Eq.~\eqref{eq:S_vN_diff}, and in fact $T\geq 1/\omega_\mathrm{min}$ is already sufficient for convergence of the entropies. Moreover, for finite dimensional systems with dense spectral energy distribution, the time average state converges slowly to equilibrium.

\subsection{\label{appendix:generalization_observational_entropy}Extension for the observational entropy}

\begin{figure*}
    \centering
    \includegraphics[width=\textwidth]{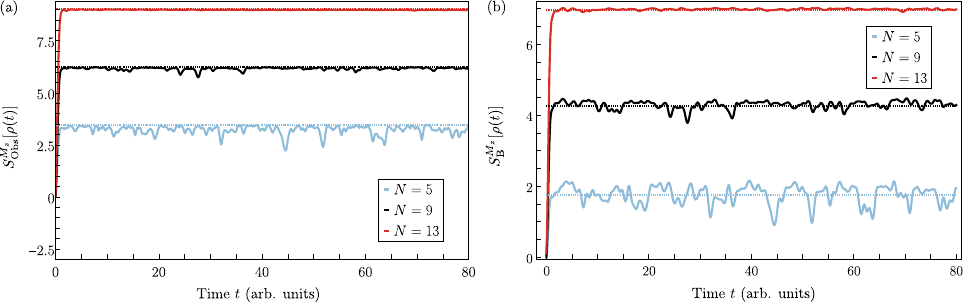}
    \caption{Here, we present exact numerical results showing two more kinds of entropy: the observational entropy in panel (a) and the Boltzmann entropy in panel (b), both defined in Eq.~\eqref{eq:BoltzEntObsEnt}. In both cases, the entropies are those associated with the system, dynamics, and operator discussed in the main text, for the same three spin-chain lengths. In panel (a), the dot-dashed lines represent the \textit{maximum possible} observational entropy for \textit{any} observable in a Hilbert space of that size:~$\log d$, where $d=\mathrm{dim}\,\mathcal{H}$. This is in contrast to panel (b) here (and Fig.~\ref{fig:spinchain_panel}(b)), where the dashed lines represent the entropies of $\omega$. As with the Shannon entropy, the observational and Boltzmann entropies also increase towards their equilibrium values for this choice of observable and initial configuration.}
    \label{fig:AppendixFigs1}
\end{figure*}

Besides the Shannon observable entropy defined in the main text, another notion of entropy relative to an observable is the observational entropy~\cite{Safranek2021,Buscemi2023}.
Given the observable with spectral decomposition $O=\sum_{i=1}^r O_i \Pi_i^O$, the observational entropy can be written as (also see the main text)
\begin{align}
    \label{eq:S_obs_appendix}
    S_\mathrm{Obs}^O[\rho] = -\sum_{i=1}^r p_i \log \frac{p_i}{V_i},
\end{align}
where again $p_i= \tr[\Pi_i^O\rho]$ is the population of $\rho$ in the $i$th eigenspace of $O$ and $V_i=\tr[\Pi_i^O]$ is the dimension of that eigenspace. The dimension $V_i$ can also be understood as the number of microstates that make up the $i$th macrostate of the observable $O$.
The observational entropy can also be understood as the entropy of the coarse-grained state $\rho^{\rm cg} = \sum_{i=1}^r p_i {\Pi_i^O}/{V_i}$ (note that this does not directly generalize to POVMs, as opposed to the remaining results)~\cite{Wehrl1978,Safranek2021},
\begin{align}
    S_{\rm vN}[\rho^{\rm cg}] &= -\sum_{i=1}^r\tr\left[p_i \frac{\Pi_i^O}{V_i} \log \frac{p_i}{V_i}\right] \\
    &= S_{\rm Obs}^O[\rho].\label{eq:Obs_entropy_CG}
\end{align}
Contrary to the Shannon observable entropy that captures the knowledge attainable purely from the measurement outcome, we can therefore understand the observational entropy as the entropy associated with the system, given access knowledge of only the coarse-grained state.

By splitting up the sum into two parts using $\log\frac{p_i}{V_i}=\log p_i - \log V_i,$ we can understand the observational entropy as a sum of the Shannon observable entropy and an averaged Boltzmann entropy,
\begin{align}
\label{eq:BoltzEntObsEnt}
    S_\mathrm{Obs}^O[\rho] = \underbrace{-\sum_{i=1}^r p_i\log p_i}_{=S_\mathrm{Sh}^O[\rho]} + \underbrace{\sum_{i=1}^r p_i \log V_i}_{=S_\mathrm{B}^O[\rho]}.
\end{align}
The first term, the Shannon observable entropy is what we have already discussed in the main text, corresponding to the entropy of the probabilistic distribution of the measurement outcomes.
The second term, the averaged Boltzmann entropy, physically corresponds to an average of the Boltzmann entropy ``$S=k_B\log W$'' over the $r$ possible macroscopic outcomes of the observable with $V_i$ the number of microscopic states compatible with the $i$th outcome of the measurement of $O$.
In Fig.~\ref{fig:AppendixFigs1} we provide a numerical example displaying how the observational entropy splits into the Shannon and Boltzmann terms.
With the additional Boltzmann term, the observational entropy further coarse grains the information available in the state $\rho(t)$, and one may therefore ask whether this additional step is fundamentally necessary to recover the second law relative to an observable.

We answer this question negatively, by first of all noting that the ratio $r/d_{\rm eff}\ll 1$ is sufficient for equilibration of both the Shannon observable entropy as well as the observational entropy (Proposition~\ref{prop:obsentropybound}).
Secondly, there are some pathological cases where the observational entropy even wrongly suggest that a system equilibrates with respect to an observable, when in fact, it does not.
One example where this happens is in in subsystem equilibration. 
For illustration, we may consider a joint spin-$\frac{1}{2}$ and bath system where the bath is assumed to have dimension $d_B$.
Suppose the spin is governed by the Hamiltonian $H_S = g\sigma_x$, it initially starts in the state $\ket{0}$, and the observable of interest is the spin polarization in $z$-direction, $O=\sigma_z\otimes \mathds 1_B$.
If the bath and the spin do not interact, the population vector for the measurement outcome evolves as $\vec p(t)=(\cos(gt)^2,\sin(gt)^2),$ and therefore $S_\mathrm{Sh}^O[\rho(t)]$ periodically fluctuates between $0$ and $\log 2$.
The Boltzmann entropy term, on the other hand, is constant, $S_B^O[\rho(t)]=\log d_B$, and in the limit of a large bath, the relative fluctuations of the Boltzmann entropy vanish as $1/\log d_B$.
Looking solely at the observational entropy, this suggests that the system equilibrates with respect to $O$, whereas physically the contrary is the case: the spin-$\frac{1}{2}$ subsystem under observation is isolated from the environment and does not equilibrate, as correctly captured by the Shannon observable entropy.

In summary, this shows that the additional coarse graining of the observational entropy can result in cases where the Shannon entropy does not equilibrate while the observational entropy does, because of the inclusion of the Boltzmann term.
Cases such as the one discussed before, however, are not physically self-consistent, because if one ignores fluctuations of order $\log 2$ from the observable, the operational way to model this is to use another observable that does not distinguish between those states, for example using a POVM instead of a PVM to model the fact that one can not differentiate between the spin orientations perfectly (see generalization in SM Ref.~\ref{appendix:replacingPOVMs}).
Observable equilibration in an operational sense therefore requires the Shannon observable entropy to equilibrate, not necessarily the observational entropy.
The reason for this apparent difference is that the observational entropy captures the entropy of the system given the information accessible via the observable (as per Eq.~\eqref{eq:Obs_entropy_CG}), whereas the Shannon observable entropy captures the information of the observable.
In the following we show that the condition $d_\mathrm{eff}\gg r$ under which the Shannon observable entropy equilibrates
is also sufficient for the observational entropy to equilibrate in a similar sense:

\begin{restatable}[Observational entropy bound]{proposition}{obsentropybound}
\label{prop:obsentropybound}
    Consider $\rho(t),$ $\omega$ and $O$ be as in Theorem~\ref{thm:shannonbound}. For any $\varepsilon > 0$ and $T> 0$, the observational entropy with respect to $O$ satisfies
    \begin{equation}
        \left\langle \left| S_\mathrm{Obs}^O[\rho(t)] - S_\mathrm{Obs}^O[\omega]\right| \right\rangle_T \leq \log(d)\eta_{\varepsilon,T} +  g\left(\eta_{\varepsilon,T}\right),
    \end{equation}
    where $g(x)=-x\log(x)+ (1+x)\log(1+x)$, with $\eta_{\varepsilon,T}$ as defined for Eq.~\eqref{eq:avg_pt_pomega}.
\end{restatable}

To prove the desired inequality for the observational entropy, we want to proceed similarly to the proof of Theorem~\ref{thm:shannonbound}.
In contrast to the Shannon entropy, however, we can not directly use inequality~\eqref{eq:entropy_continuity_bound} to relate the difference in observational entropies to differences in the populations relative to the observable $O$.
The following Lemma, which is a tightening of the continuity bound in Theorem~6 from Ref.~\cite{schindlerwinter2023}, provides the desired statement.
\begin{restatable}[Observational entropy continuity bound]{lemma}{obsentropycontinuity}\label{lemma:obsentropycontinuity}
    Let $\rho$ and $\sigma$ be two quantum states defined on a $d$-dimensional Hilbert space, and have $O$ be an observable on the same space.
    Define $\vec p_\rho$ and $\vec p_\sigma$ to be the population vectors of $\rho$ and $\sigma$, respectively, relative to $O$.
    Then,
    \begin{align}
    \label{eq:obsentropy_continuity}
        |S_\mathrm{Obs}^O[\rho]-S_\mathrm{Obs}^O[\sigma]|\leq g(\|\vec p_\rho - \vec p_\sigma \|_1) + \log d \, \|\vec p_\rho - \vec p_\sigma \|_1,
    \end{align}
    where $g(x)=-x\log(x) + (1+x)\log(1+x)$ as before, and the 1-norm $\|\cdot\|_1$ for vectors is defined as in the main text (or equivalently Eq.~\eqref{eq:1-norm_vectors} of the Appendix).
\end{restatable}
\begin{proof}
    Let us start by considering a function $Z[\rho]$ defined on finite dimensional quantum states $\rho$ which is bounded concave.
    Bounded convacity means that for any convex combination $\rho=\sum_k \lambda_k \rho_k,$ with $\lambda_k\geq 0,$ and $\sum_k\lambda_k=1$ and $\rho_k$ a well-defined quantum state, we have~\cite{schindlerwinter2023}
    \begin{align}
        0\leq Z\left[\rho\right] - \sum_k\lambda_k Z[\rho_k] \leq S_{\rm Sh}[\{\lambda_k\}_k],
    \end{align}
    where $S_{\rm Sh}[\{\lambda_k\}_k]=-\sum_k\lambda_k\log \lambda_k$ is the Shannon entropy of the distribution $\{\lambda_k\}_k$.
    It has been shown in~\cite{schindlerwinter2023}, that any such function satisfies the inequality
    \begin{align}
    \label{eq:Z_inequality}
        |Z[\rho] - Z[\sigma]|\leq \kappa \|\rho-\sigma\|_1 + g(\|\rho-\sigma\|_1),
    \end{align}
    where $\kappa := \sup_{\mu,\nu}|Z[\mu]-Z[\nu]|$.
    Note the difference in normalization of the trace norm in our work Eq.~\eqref{eq:1-norm-quantum} to that of~\cite{schindlerwinter2023}.
    For our purposes let us consider an $r$-dimensional state-space where we identify the population vector $\vec p_\rho$ with a quantum state
    \begin{align}
    \label{eq:quantified_vector}
        \vec p_\rho = \sum_{i=1}^r \tr\left[\Pi_i^O \rho \right] \ketbra{i}{i},
    \end{align}
    and similarly for $\vec p_\sigma$.
    We now define $Z[\,\cdot\,]$ on this $r$-dimensional state space by setting
    \begin{align}
    \label{eq:Z_def}
        Z[\tau] := -\sum_{i=1}^r \tau_{ii} \log \frac{\tau_{ii}}{V_i},
    \end{align}
    where $\tau_{ii}$ is the diagonal entry of the $r\times r$ quantum state $\tau$.
    Inserting the diagonal state $\vec p_\rho,$ we recover the observational entropy $Z[\vec p_\rho] = S_{\rm Obs}^O[\rho]$.
    We note that $Z$ is bounded concave: let $\tau = \sum_k\lambda_k \tau_k$ be a finite convex combination, then
    \begin{widetext}
    \begin{align}
        Z[\tau] - \sum_k \lambda_k Z[\tau_k] &= S_{\rm Sh}[\{\tau_{ii}\}_i] + \sum_{i=1}^r \tau_i \log V_i - \sum_k \lambda_k S_{\rm Sh}[\{\tau_{k,ii}\}_i] - \sum_k \sum_{i=1}^r\lambda_k \tau_{k,ii} \log V_i \\
        &= S_{\rm Sh}[\{\tau_{ii}\}_i] - \sum_k \lambda_k S_{\rm Sh}[\{\tau_{k,ii}\}_i],
    \end{align}
    \end{widetext}
    is equal the difference of the diagonal Shannon entropies. The Shannon entropy is bounded concave by Lemma~1 of Ref.~\cite{schindlerwinter2023}, and thus, also $Z$ is.
    Therefore, our definition~\eqref{eq:Z_def} satisfies all the necessary assumptions for the inequality~\eqref{eq:Z_inequality} to hold.
    By applying this inequality to the diagonal states $\vec p_\sigma$ and $\vec p_\rho$ we find that
    \begin{align}
        |Z[\vec p_\rho] - Z[\vec p_\sigma]|\leq \kappa \|\vec p_\rho- \vec p_\sigma\|_1 + g(\|\vec p_\rho- \vec p_\sigma\|_1).
    \end{align}
    Note that because $\vec p_\rho$ and $\vec p_\sigma$ are diagonal in the same basis, the 1-norm for vectors defined in Eq.~\eqref{eq:1-norm_vectors} and the 1-norm for states defined in Eq.~\eqref{eq:1-norm-quantum} agree. Using that $\kappa = \log d$ for the observational entropy, we conclude the proof.
    
\end{proof}
Note that this result, heavily based on the work in Ref.~\cite{schindlerwinter2023}, is a tightening of Theorem~6 from the same reference.
The underlying reason is that the trace-distance between two quantum states $\rho$ and $\sigma$ is always lower bounded by the 1-norm distance of their populations relative to any observable~\cite{Nielsen2010},
\begin{align}
    \|\rho-\sigma\|_1 \geq \|\vec p_\rho - \vec p_\omega\|_1.
\end{align}
Thus, the righthand side of~\eqref{eq:obsentropy_continuity} is always smaller that of Theorem~6 from~\cite{schindlerwinter2023}.
With this bound at hand, we are ready to prove Proposition~\ref{prop:obsentropybound}.

\begin{proof}[Proof of Proposition~\ref{prop:obsentropybound}.]
    Applying Lemma~\ref{lemma:obsentropycontinuity} to this setting, we find that
    \begin{align}
        &\Big\langle \big| S_\mathrm{Obs}^O[\rho(t)] - S_\mathrm{Obs}^O[\omega]\big| \Big\rangle_T \\
        &\quad \leq \Big\langle \log d\, \|\vec p(t)- \vec p_\omega\|_1 + g\left(\|\vec p(t) - \vec p_\omega\|_1\right)\Big\rangle_T \\
        &\quad \leq \log d\, \eta_{\varepsilon,T} +  g\left(\eta_{\varepsilon,T}\right),
 \end{align}
where in the second line, we use the concavity and monotonicity of $g(x) = -x\log x + (1+x)\log(1+x)$, proving the statement.
\end{proof}

The origin for the $\log d$ scaling comes from the Boltzmann term which scales at worst with $\log d$, the maximum number of microstates in the system. However, due to the additional factor of $\eta_{\varepsilon,T}\propto \sqrt{{r}/{d_\mathrm{eff}}}$, if the effective dimension is of the order of the dimension of the system $d_\mathrm{eff} \sim d$, this additional correction will also vanish in the limit of large systems $d_\mathrm{eff}\gg r$.
In a very close analogy to Corollary~\ref{cor:shannonfluctuations}, the observational entropy also satisfies a fluctuation theorem, bounding the probability of small fluctuations of $S_\mathrm{Obs}^O[\rho(t)]$ away from the equilibrium value $S_\mathrm{Obs}^O[\omega]$ as we show below:

\begin{restatable}[Observable entropy fluctuations]{corollary}{observablefluctuations}
    \label{cor:observablefluctuations}
    Let $\rho(t), \rho^\infty$ and $\hat O$ be as in Lemma~\ref{lemma:populationbound}.
    If $t\in \mathbb R_{\geq 0}$ is sampled uniformly randomly, we have
    \begin{align}
        P\left[\left| S_\mathrm{Obs}^O[\rho(t)] - S_\mathrm{Obs}^O[\omega]\right|\geq \sqrt{\nu}\right]\leq\sqrt{\nu},
    \end{align}
    where $\nu$ is given by
    \begin{align}
        \nu = \frac{\log(d)}{2}\sqrt{\frac{r}{d_\mathrm{eff}}} + g\left[\frac{1}{2}\sqrt{\frac{r}{d_\mathrm{eff}}}\right],
    \end{align}
    with $g$ the function as in Proposition~\ref{prop:obsentropybound}.
\end{restatable}
\begin{proof}[Proof sketch.]
    The proof goes exactly as that of Corollary~\ref{cor:observablefluctuations} with the difference that we define $X=\left| S_\mathrm{Obs}^O[\rho(t)] - S_\mathrm{Obs}^O[\omega]\right|$ for Chebyshev's inequality. The reader is referred to the proof of Corollary~\ref{cor:shannonfluctuations} in Sec.~\ref{appendix:proofs_entropy_bounds} for details.
\end{proof}
Akin to the generalization in Eq.~\eqref{eq:generalized_fluctuations_shannon} of the fluctuation bounds for the Shannon observable entropy from Corollary~\ref{cor:shannonfluctuations}, we can also generalize Corollary~\ref{cor:observablefluctuations} for the case of the observational entropy,
\begin{align}
    P\left[\left|S_\mathrm{Obs}^O[\rho(t)] - S_\mathrm{Obs}^O[\omega] \right|\geq \mu \right]\leq \frac{\nu}{\mu}.
\end{align}
As before, $\mu$ is the small parameter relative to which we compare the observational entropy fluctuations on the left-hand side.

\subsection{\label{appendix:diagonal_entropy}On the diagonal entropy}
In the context of the diagonal entropy as for example used in the Refs.~\cite{Polkovnikov2011,Ikeda2015}, a partial extension of our results is also possible.
The mentioned references define the diagonal entropy relative to the orthonormal energy eigenbasis $\{\ket{\varepsilon_n}\}_{1\leq n\leq d}$ of the system's Hamiltonian $H$ assumed to be non-degenerate.
For a quantum state $\rho$ the diagonal entropy is then given by 
\begin{align}
\label{eq:S_d}
    S_{\rm d}[\rho]=-\sum_{n=1}^d \rho_{nn}\log\rho_{nn},
\end{align}
where $\rho_{nn}=\braket{\varepsilon_n}{\rho|\varepsilon_n}$ are the diagonal entries of $\rho$ relative to the energy eigenbasis.
In case of an isolated quantum system, however, the Hamiltonian is time-independent, and consequently also the diagonals $\rho_{nn}(t)$ of the state $\rho(t)=e^{-iHt}\rho(0)e^{iHt}$.
The diagonal entropy relative to the Hamiltonian $S_{\rm d}[\rho(t)]$ thus trivializes to a constant and does not quantify whether and how an isolated quantum system equilibrates.

When defining a diagonal entropy more generally relative to some orthonormal basis $\mathcal B = \{\ket{n}\}_{1\leq n\leq d}$ of the Hilbert space of interest allows for more informative conclusions.
The diagonal elements $\rho_{nn}=\braket{n}{\rho|n}$ of $\rho$ can be taken in the basis representation of the states $\ket{n}$ instead of $\ket{\varepsilon_n}$.
A diagonal entropy $S_{\rm d}^{\mathcal B}$ relative to $\mathcal B$ can then be defined akin to the diagonal entropy in Eq.~\eqref{eq:S_d}.
Since the diagonal vector of $\rho$ is in general $d$-dimensional, Lemma~\ref{lemma:populationbound} does not directly constrain the diagonals because in this case $r=d$ does not satisfy $r\ll d_{\rm eff}$.
However, when considering an observable $O$ with $r$ outcomes satisfying $r\ll d_{\rm eff}$ as in the main text, a preferred basis can be defined as in~\cite{Scarpa2023}.
Relative to the observable basis, the diagonal entropy can satisfy a second law in a sense similar to Theorem~\ref{thm:shannonbound}.
Decomposing the eigenprojectors of the observable,
\begin{align}
    \Pi_i^O = \sum_{k=1}^{V_i} \ketbra{i,k}{i,k},
\end{align}
into a sum of $V_i=\tr[\Pi_i^O]$ 1-dimensional orthogonal projectors $\ketbra{i,k}{i,k}$ defines a basis $\mathcal B_O = \{\ket{i,k}\}_{i,k}$.
Within each subspace $1\leq i \leq r$ the choice of basis is not uniquely defined by $O$, and arbitrary unitary transformations within the subspace $\{\ket{i,k}\}_k$ for fixed $i$ map the basis to another valid orthonormal one.
Chosing the basis such that the diagonal entropy $S_{\rm d}[\rho]$ is maximized yields matrix elements,
\begin{align}
    \braket{i,k}{\rho |i,k} = \braket{i,k'}{\rho |i,k'}=\frac{p_i}{V_i},
\end{align}
which are equal within each subspace defined by the observable's eigenspace projectors $\Pi_i^O$.
In this case, the diagonal entropy equals the observational entropy defined by Eq.~\eqref{eq:obs entropy}, as shown in~\cite{Scarpa2023},
\begin{align}
    \max_{\mathcal B_O}S_{\rm d}[\rho] &= -\min_{\mathcal B_O}\sum_{i=1}^r \sum_{k=1}^{V_i} \frac{p_i}{V_i}\log \frac{p_i}{V_i}\\ 
    &= S_{\rm Obs}^O[\rho].
\end{align}
As we have shown in Appendix~\ref{appendix:generalization_observational_entropy}, this entropy also satisfies a second law bound like the Shannon observable entropy and thus also the diagonal entropy defined in the context of sufficiently coarse-graining observable.

\subsection{\label{appendix:thermodynamic_entropy}Special case of thermodynamic entropy}
When comparing our formulation of the second law based on Theorem~\ref{thm:shannonbound} with the historical ones by Clausius~\cite{Clausius1854} or Kelvin-Planck~\cite{Planck1897} certain additional assumptions have to be imposed on the setting under consideration.
These classical versions of the second law are framed through the concept of thermalization.
By contrast, the present work adopts the broader notion of equilibration where the focus is on how observables evolve toward stable equilibrium values.
To recover the standard thermodynamic notions of entropy from our results, two separate concepts have to be introduced: (a) `General canonical typicality'~\cite{Popescu2006,linden2009quantum} stating that most states of interest are close to a specific canonical state $\Omega_S$ in a sense which we specify in the following, and (b) the `thermal canonical principle' stating that the canonical state $\Omega_S$ is close to the thermal state under appropriate assumptions~\cite{Landau1980,Goldstein2006}.

In this appendix, we argue how results on thermalization from previous works like Refs.~\cite{Goldstein2006,Popescu2006,linden2009quantum} and standard techniques from statistical mechanics can be combined with our main Theorem~\ref{thm:shannonbound} to recover the thermodynamic entropy as a special case of the Shannon observable entropy.
To comprehensively arrive at this goal, we use the canonical ensemble where the isolated system is made up of two subsystems, one of which we call the system $S$ and the other the environment $E$.
In the non-trivial case where $S$ and $E$ interact, $S$ is not isolated anymore, only the joint global system is.

Ultimately, we would wish to find that under reasonable assumptions which are to be specified, an initial state $\ket{\psi}$ on the global system yields a thermal equilibrium state on the system $S$. 
Being more precise, we wish to find that the equilibrium state $\omega_\psi = \langle \ketbra{\psi}{\psi}\rangle_\infty$ coming from $\ket{\psi}$, reduced on the system $S$ is close to the thermal Gibbs state at some inverse temperature $\beta$, meaning
\begin{align}
\label{eq:tau_S_approx}
    \tr_E[\omega_\psi] \approx \tau_S = \frac{e^{-\beta H_S}}{Z_S},
\end{align}
where $H_S$ is the system Hamiltonian and $Z_S=\tr[e^{-\beta H_S}]$ is the partition function.
Let us suppose for now that the relation~\eqref{eq:tau_S_approx} is satisfied; we return to this assumption later.
Furthermore, we consider the system's Hamiltonian $H_S$ as the observable of interest, $O=H_S\otimes \mathds 1_E$.
The Shannon entropy relative to $O$ evaluated at the equilibrium state $\omega_\psi$ then yields the usual thermodynamic entropy.
By explicitly calculating 
\begin{align}
    p_i =\tr\left[\left(\Pi_i^{H_S}\otimes\mathds 1_E\right) \omega_\psi\right]
        =\frac{e^{-\beta \varepsilon_i}}{Z_S},
\end{align}
where $\Pi_i^{H_S}$ is the projector on the $i$th energy eigenstate of $H_S$, the Shannon observable entropy can be determined to be
\begin{align}
    S_{\rm Sh}^{H_S}[\omega_\psi] &= -\sum_{i=1}^{d_S} p_i \log p_i \\
    &= \log Z_S + \beta\langle H_S\rangle.
\end{align}
Thus, assuming equality for~\eqref{eq:tau_S_approx}, the Shannon observable entropy of the equilibrium state relative to the system's Hamiltonian as the observable is in agreement with the thermodynamic entropy.

In practice, however, the relation~\eqref{eq:tau_S_approx} is only true approximately. 
Moreover, even though the deviations $\kappa:=\|\tr_E[\omega_\psi]-\tau_S\|_1$ are small for most states, there are some where $\kappa$ significantly deviates from zero.
It is also known that not all systems thermalize and thus quantifying these deviations has been the subject of extensive studies.
In the remainder of this appendix, we provide pertinent references for how and under which assumptions the Gibbs state can be recovered.
In our setting of interest, the total Hamiltonian is then given by
\begin{align}
    H = H_S\otimes \mathds 1_E + \mathds 1_S\otimes H_E + H_{SE},
\end{align}
where $H_S$ is the Hamiltonian of $S$ as before, and $H_E$ is the environment Hamiltonian $H_E$ and $H_{SE}$ are the interactions, assumed to be negligible compared to the other two terms.
Conventionally, the initial state $\ket{\psi}$ is assumed to be a uniformly randomly picked state from a macroscopic slice $\mathcal H_R\subseteq \mathcal H_S\otimes\mathcal H_E$ from the full Hilbert space~\cite{Popescu2006}.
To recover the Gibbs state as in e.g.~\cite{Goldstein2006}, the microcanonical slice is considered, which is spanned by all energy eigenstates of $H$ with energy eigenvalue within the window $[E,E+ \Delta)$.
Conventionally, it is assumed that $\Delta$ is small (on the scale of $E$ minus the ground state energy) but large enough that there is a macroscopic number of energy eigenstates within the window $E\pm \Delta$.

Starting with the first point (a), the canonical state $\Omega_S$ of the system $S$ is the one obtained by uniformly randomly picking a state from the macroscopic slice $\mathcal H_R$ and then reducing to $S$,
\begin{align}
\label{eq:Omega_S}
    \Omega_S = \tr_E \left[\frac{\mathds 1_R}{d_R}\right].
\end{align}
The `general canonical principle' as in Popescu~et~al.~\cite{Popescu2006} then states that most states $\ket{\psi}\in\mathcal H_R$ will yield a reduced state on the system very close to the canonical state if the environment $E$ is sufficiently larger than the system $S$.
Formally, the general canonical principle quantifies the average distance $\langle\|\tr_E[\ketbra{\psi}{\psi}_E]-\Omega_S\|_1\rangle_\psi$ where the average $\langle\cdot\rangle_\psi$ is the Haar average over the subspace $\mathcal H_R$.
This is however, not yet enough, as we wish to have that also the equilibrium state $\tr_E[\omega_\psi]$ reduced on the subsystem is close to the canonical state.
In the reference by Linden~et~al.~\cite{linden2009quantum}, such a statement is shown, providing a bound for the average distance $\langle \|\tr_E[\omega_\psi]-\Omega_S\|_1\rangle_\psi$.
Note that there is a technical difference in how Popescu~et~al.\ and Linden~et~al.\ define the canonical state; while Popescu~et~al.\ define $\Omega_S$ according to eq.~\eqref{eq:Omega_S}, Linden~et~al.\ define it as $\tilde\Omega_S = \langle \tr_E[\omega_\psi] \rangle_\psi$.
Due to linearity of the trace and the map $\ketbra{\psi}{\psi}\mapsto\omega_\psi$, both definitions agree.

The second point (b) is about the `thermal canonical principle' and a statement about when canonical state $\Omega_S$ equals the Gibbs state $e^{-\beta H_S}/Z_S$ when we consider the microcanonical slice $\mathcal H_R$.
This problem can be addressed using textbook methods from e.g.\ Ref.~\cite{Landau1980}, where a standard assumption is that the interactions $H_{SE}$ are weak compared to the rest of the Hamiltonian.
Such an assumption not needed for the more general canonical principle~\cite{Popescu2006,linden2009quantum}.
In the reference~\cite{Goldstein2006}, this special case is considered more rigorously by providing an exemplary derivation for how the microcanonical slice $\mathcal H_R$ yields the Gibbs state, given that the Hamiltonian spectrum is sufficiently dense (yielding a macroscopic slice $\mathcal H_R$).

In summary, we find that by combining results from the literature on typicality~\cite{Popescu2006,linden2009quantum} with thermal canonical typicality~\cite{Landau1980,Goldstein2006} from statistical mechanics, the Shannon observable entropy of a suitably chosen (sub)system coincides with the standard thermodynamic notion of entropy from the canonical ensemble.

\section{\label{appendix:replacingPOVMs}Replacing projective observables with POVMs}
So far, we have discussed observable-relative notions of entropy.
Instead of looking at an observable corresponding to the operator $O=\sum_{i=1}^r O_i \Pi_i^O $ and defining the measurement outcomes relative to $O$, we can instead work with the more general notion of a positive operator-valued measure (POVM)~\cite{Nielsen2010}.
A POVM with $r$ outcomes is defined by a set of bounded operators $M = (E_i)_{i=1}^r$ on the Hilbert space of our system. Collectively, the operators in the set must satisfy  $\sum_{i=1}^r E_i = \mathds 1$.
The outcome of a generalized measurement over $\rho$ is $i$ with the probability of obtaining outome $i$ given by $p_i=\tr[E_i\rho]$. Therefore, we can generalize the Shannon observable entropy to a POVM using the following expression
\begin{align}
    S_\mathrm{Sh}^M[\rho] = -\sum_{i=1}^r p_i \log p_i,\quad p_i :=\tr[E_i\rho],
\end{align}
defined analogously to the Shannon entropy of the probability vector $\vec p = (p_1,\dots,p_r)$ relative to the POVM $M$.
The physical interpretation remains the same for this entropy as for the notion defined relative to an observable; the entropy $S_\mathrm{Sh}^M[\rho]$ is a measure of the information relative to $M$ that one gains by performing a measurement.
It is important to note, however, that some of the entropy may when using POVMs may not be due to an intrinsic uncertainty coming from the underlying quantum state as is the case for PVMs, but some of the uncertainty comes from the coarse graining of the POVM. An examplary case is if $M=(E_1,E_2)$ and $E_1=E_2=\mathds 1/2,$ the Shannon observable entropy is $\log 2$, not due to the underlying state but due to the randomness coming from the POVM.
Similarly, the observational entropy can also be generalized using
\begin{align}
    S_\mathrm{Obs}^M[\rho] = -\sum_{i=1}^r p_i\log \frac{p_i}{V_i},\quad p_i :=\tr[E_i\rho],\, V_i = \tr[E_i],
\end{align}
where $V_i$ is the generalized number of microstates compatible with the outcome $i$ of the measurement $M$. The condition $\sum_{i=1}^r E_i = \mathds 1$ ensures that $\sum_{i=1}^r V_i=\mathrm{dim}\,\mathcal H$ equals the total number of microstates of the system, i.e.\ the dimension of the Hilbert space on which $\rho$ is defined.

In the main text we looked at how the entropy of $\rho(t)$ relative to an observable $O$ compares to the entropy of $\omega$ (relative to the same observable $O$).
We wish to generalize these statements to the case where the entropy is now defined relative to the POVM $M$ in question.
Since our main results Theorem~\ref{thm:shannonbound}, Proposition~\ref{prop:obsentropybound} and Corollary~\ref{cor:shannonfluctuations} all follow from the asymptotic closeness of the population vector $\vec p(t)$ to the equilibrium populations $\vec p_\infty$, what must be shown is that for a more general POVM $M$, the two population vectors are also close for most times.
It turns out that Lemma~\ref{lemma:populationbound} also generalizes to the case of POVMs:

\begin{restatable}[POVM population equilibration]{lemma}{POVMpopbound}
\label{lemma:POVMpopbound}
    Let $M=(E_i)_{i=1}^r$ be a POVM with $r$ outcomes, and let $\rho(t)$ be the unitary time evolution of some initial state $\rho_0$ with respect to the Hamiltonian $H$.
    Let $\vec p(t)$ and $\vec p_\omega$ be defined with respect to $M$ with entries $p_i = \tr[E_i\rho]$.
    Furthermore, take $\varepsilon>0$ and $T>0$ to be arbitrary, then,
    \begin{align}
        \left\langle \left\|\vec p(t) - \vec p_\infty \right\|_1\right\rangle_T \leq \frac{1}{2}\sqrt{\frac{r}{d_\mathrm{eff}}f(\varepsilon,T)} =: \eta_{\varepsilon,T},
    \end{align}
    where the constant $\eta_{\varepsilon,T}$ is defined with $f(\varepsilon,T)$ as in Eq.~\eqref{eq:f_eps_T} of the main text.
\end{restatable}
\begin{proof}
The proof is similar to that of Lemma~\ref{lemma:populationbound}, with the difference that while for projectors we can use the fact that $\Pi_i^2 = \Pi_i$, this is generally not the case for a POVM: $E_i^2 \neq E_i$.
Therefore, we have to replace Eq.~\eqref{eq:Pi_omega2} with
\begin{align}
    \left\langle\left\|\vec p(t) - \vec p_\omega\right\|_2^{\,2}\right\rangle_T \leq f(\varepsilon,T)\sum_{i=1}^r \tr\left[E_i^2 \omega^2\right],
\end{align}
where we have already used the hermiticity of the POVM elements $E_i^\dagger = E_i$.
Furthermore, we can split
\begin{align}
    \tr\left[E_i^2 \omega^2\right] = \tr\left[E_i\sigma_i\right]\tr\left[E_i\omega^2\right],
\end{align}
where we defined $\sigma_i = (M_i\omega^2 M_i)/\tr\left[E_i\omega^2\right]$ as the post-measurement state, with $M_i=\sqrt{E_i}$ the measurement operator associated to the POVM element $E_i$.
Then, we can use that $\tr[E_i\sigma_i]\leq 1$ (since they are probabilities) to arrive the desired result, as in Eq.~\eqref{eq:omega2_deff},
\begin{align}
    \left\langle\left\|\vec p(t) - \vec p_\omega\right\|_2^{\,2}\right\rangle_T \leq f(\varepsilon,T)\sum_{i=1}^r \tr\left[E_i \omega^2\right] = \frac{f(\varepsilon,T)}{d_\mathrm{eff}},
\end{align}
again using the completeness relation $\sum_{i=1}^r E_i = \mathds 1.$
\end{proof}

As a consequence, the resulting inequalities from Theorem~\ref{thm:shannonbound} for the Shannon entropy $S_\mathrm{Sh}^M$ and from Proposition~\ref{prop:obsentropybound} for the observational entropy $S_\mathrm{Obs}^M$ also follow for the POVM $M$.
For the Shannon entropy bound, one can directly use Lemma~\ref{lemma:POVMpopbound} instead of Lemma~\ref{lemma:populationbound} to prove the statement.
On the other hand, for the observational entropy, the proof of Lemma~\ref{lemma:obsentropycontinuity} has to be modified to also hold for a general POVM $M$, instead of using the definition with respect to the orthonormal projectors that we have for an observable $O$.
The proof goes through completely analogously simply by replacing eq.~\eqref{eq:quantified_vector} with 
\begin{align}
    \vec p_\rho = \sum_{i=1}^r \tr[E_i \rho]\ketbra{i}{i},
\end{align}
thus generalizing both Lemma~\ref{lemma:obsentropycontinuity} and Proposition~\ref{prop:obsentropybound} to the case of POVMs.
A useful consequence of this generalization is that with a POVM one can model measurements which are coarse-grained in the sense that they can be used to model measurements which do not perfectly distinguish some states.

\section{\label{sec:details_numerics_model}Details of the numerical model used}
Here, we describe in more detail the model used to generate the numerical results in Figs.~\ref{fig:spinchain_panel}, \ref{fig:AppendixFigs1} and~\ref{fig:AppendixFigs2}.
The model used is given in Ref.~\cite{13KimHuse} with the Hamiltonian
\begin{align}
    H = \sum_{i=1}^N g \sigma_x^{(i)} &+ \sum_{i=2}^{N-1}h\sigma_z^{(i)} + (h-J)\left(\sigma_z^{(1)}+\sigma_z^{(N)}\right) \nonumber \\
    &+ \sum_{i=1}^{N-1} J \sigma_z^{(i)}\sigma_z^{(i+1)},
\end{align}
where $\sigma_x^{(i)}$ and $\sigma_z^{(i)}$ are the Pauli operators acting on the spin on site $i$.
The number $N$ is the total number of spins in the chain, and the remaining constants have been chosen as in Ref.~\cite{13KimHuse} to be $h=(\sqrt{5}+1)/4=0.8090...$, $g=(\sqrt{5}+5)/8=0.9045...$ and $J=1$ (recalling that $\hbar=1$ as always). The initial state used in the generation of all plots is $\ket{\downarrow\dots\downarrow}$, with each spin assigned a random phase. The observable $M_y$ is defined analogously to $M_z$: $M_y=\frac{1}{N}\sum_i^N \sigma_y^{\left(i\right)}$. The time evolution of the initial state was calculated by exact diagonalization of the Hamiltonian to construct a time-evolution unitary separately at each time step. The Python code used to generate the data is available upon request. The time-averaging in Fig.~\ref{fig:spinchain_panel}(c) was performed numerically using Mathematica's Interpolation and NIntegrate functions, and the exponential fitting in Fig.~\ref{fig:spinchain_panel}(d) was performed using Mathematica's NonLinearModelFit function.

\bibliography{bibfile.bib}

\end{document}